\documentclass[12pt,draftcls,onecolumn]{IEEEtran}

\usepackage{cite}
\usepackage{graphicx}
\usepackage{amsmath,amssymb,amsfonts}
\usepackage{array}
\usepackage{flushend}
\usepackage{color}

\makeatletter%
\if@twocolumn%
\newcommand{\Figwidth}{0.9\columnwidth}%
\def\twocolbreak{\nonumber\\ &}%
\def\twocolnewline{\nonumber\\}%
\def\twocolAlignMarker{&}%
\else
\newcommand{\Figwidth}{4.5in}%
\def\twocolbreak{}%
\def\twocolnewline{}%
\def\twocolAlignMarker{}%
\fi%
\makeatother%

\allowdisplaybreaks

\begin{document}
\title{Degrees of Freedom of the Broadcast Channel with Hybrid CSI at Transmitter and Receivers}

\author{Mohamed Fadel, {\em Student Member, IEEE}, and Aria Nosratinia, {\em Fellow, IEEE}
\thanks{The authors are with the Department of Electrical Engineering, University of Texas at Dallas, Richardson, TX 75083-0688 USA, E-mail: mohamed.fadel@utdallas.edu;aria@utdallas.edu. This work was presented in part at the IEEE International Symposium on Information Theory (ISIT), Germany, June 2017 \cite{Fadel_block}.}
}
\maketitle

\newtheorem{theorem}{Theorem}
\newtheorem{lemma}{Lemma}
\newtheorem{remark}{Remark}

\def\A{\mathbf A}
\def\Aset{\mathcal A}
\def\a{\mathbf a}
\def\B{\mathbf B}
\def\b{\mathbf b}
\def\D{\mathbf D}
\def\Dset{\mathcal D}
\def\d{\mathbf d}
\def\E{E}
\def\e{\mathbf e}
\def\F{\mathbf F}
\def\f{\mathbf f}
\def\G{\mathbf G}
\def\g{\mathbf g}
\def\Gset{{\mathcal G}}
\def\H{\mathbf H}
\def\h{\mathbf h}
\def\Hset{\mathcal H}
\def\I{\mathbf I}
\def\J{\mathbf J}
\def\Jset{\mathcal J}
\def\K{ K}
\def\Kset{\mathcal K}
\def\L{L}
\def\l{\ell}
\def\N{N}
\def\m{m}
\def\M{M}
\def\o{p}
\def\P{\mathbf P}
\def\Pset{\mathcal P}
\def\p{p}
\def\Q{Q}
\def\R{R}
\def\q{q}
\def\r{r}
\def\S{\mathcal S}
\def\T{T}
\def\Tset{\mathcal T}
\def\U{\mathbf U}
\def\u{\mathbf u}
\def\Uset{\mathcal U}
\def\V{\mathbf V}
\def\v{\mathbf v}
\def\W{\mathbf W}
\def\w{\mathbf w}
\def\X{\mathbf X}
\def\x{\mathbf x}
\def\Y{\mathbf Y}
\def\Yp{Y'}
\def\y{\mathbf y}
\def\Z{\mathbf Z}
\def\z{\mathbf z}
\def\bigO{O}
\def\ZeroMat{\mathbf 0}
\def\Onevec{\mathbf 1}
\def\XiNoise{\mathbf N}
\def\Cov{\mathbf \Sigma}


\begin{abstract}
In general the different links of a broadcast channel  may experience different fading dynamics and, potentially, unequal or {\em hybrid} channel state information (CSI) conditions. The faster the fading and the shorter the fading block length,  the more often the link needs to be trained and estimated at the receiver, and the more likely that CSI is stale or unavailable at the transmitter. Disparity of link fading dynamics in the presence of CSI limitations can be modeled by a multi-user broadcast channel with both non-identical link fading block lengths as well as dissimilar link CSIR/CSIT conditions.
This paper investigates a MISO broadcast channel where some receivers experience longer coherence intervals (static receivers) and have CSIR, while some other receivers experience shorter coherence intervals (dynamic receivers) and do not enjoy free CSIR. We consider a variety of CSIT conditions for the above mentioned model, including  no CSIT, delayed CSIT, or hybrid CSIT. To investigate the degrees of freedom region, we employ interference alignment and beamforming along with a product superposition that allows simultaneous but non-contaminating transmission of pilots and data to different receivers. Outer bounds employ the extremal entropy inequality as well as a bounding of the performance of a discrete memoryless multiuser {\em multilevel} broadcast channel. For several cases, inner and outer bounds are established that either partially meet, or the gap diminishes with increasing coherence times.
\end{abstract}

\begin{keywords}
Broadcast channel, Channel state information, Coherence time, Coherence diversity, Degrees of freedom, Fading channel, Multilevel broadcast channel, Product superposition.
\end{keywords}
\IEEEpeerreviewmaketitle

\section{Introduction}


The performance of a broadcast channel depends on both the channel dynamics as well as the availability and the quality of channel state information (CSI) on the two ends of each link~\cite{Huang_degrees,Vaze_degree,Lapidoth_capacity,Jafar_blind}. The two issues of CSI and the channel dynamics are practically related. The faster the fading, the more often the channel needs training, thus consuming more channel resources, while a very slow fading link requires infrequent training, therefore slow fading models often assume that CSIR is available due to the cost of training being small when amortized over time.


In practice, in a broadcast channel some links may fade faster or slower than others. Recently, it has been shown~\cite{Fadel_broadcast,Fadel_disparity}, that the degrees of freedom of the broadcast channel are affected by the disparity of link fading speeds, but existing studies have focused on a few simple and uniform CSI conditions, e.g., neither CSIT nor CSIR were available in~\cite{Fadel_broadcast,Fadel_disparity} for any user. This paper studies a broadcast channel where the links experience both disparate fading conditions as well as non-uniform or {\em hybrid} CSI conditions. 


A review of the relevant literature is as follows. Under perfect instantaneous CSI, the degrees of freedom of a broadcast channel increase with the minimum of the transmit antennas and the total number of receivers antennas~\cite{Caire_achievable,Weingarten_capacity}. However, due to the time-varying nature of the channel and feedback impairments, perfect instantaneous transmit-side CSI (CSIT) may not be available, and also receive-side CSI (CSIR) can be assumed for slow-fading channels only.

Broadcast channel with perfect CSIR has been investigated under a variety of CSIT conditions, including imperfect, delayed, or no CSIT~\cite{Huang_degrees,Vaze_degree,Lapidoth_capacity,Davoodi_aligned,Maddah_completely,Gou_optimal}. 
In the absence of CSIT, Huang {\em et al.}~\cite{Huang_degrees} and Vaze and Varanasi~\cite{Vaze_degree} showed that the degrees of freedom collapse to that of the single-receiver, since the receivers are {\em stochastically equivalent} with respect to the transmitter. For a MISO broadcast channel Lapidoth {\em et al.}~\cite{Lapidoth_capacity} conjectured that as long as the precision of CSIT is finite, the degrees of freedom collapse to unity. This conjecture was recently settled in the positive by Davoodi and Jafar in~\cite{Davoodi_aligned}. Moreover, for a MISO broadcast channel under perfect delayed CSIT Maddah-Ali and Tse in~\cite{Maddah_completely} showed using retrospective interference alignment that the degrees of freedom are $\frac{1}{ 1 + \frac{1}{2} + \ldots + \frac{1}{K}} > 1$, where $K$ is the number of the transmit antennas and also the number of receivers. A scenario of mixed CSIT was investigated in~\cite{Gou_optimal}, where the transmitter has partial knowledge on the current channel in addition to delayed CSI. 

The potential variation between the quality of feedback links has led to the model of hybrid CSIT, where the  CSIT with respect to different links may not be identical~\cite{Tandon_fading,Amuru_degrees,Davoodi_aligned,Tandon_synergistic}. A MISO broadcast channel with perfect CSIT for some receivers and delayed for the others was studied by Tandon {\em et al.} in~\cite{Tandon_fading} and Amuru {\em et al.} in~\cite{Amuru_degrees}.  Davoodi and Jafar in~\cite{Davoodi_aligned} showed that for a MISO two-receiver broadcast channel under perfect CSIT for one user and no CSIT for the other, the degrees of freedom collapse to unity. Tandon {\em et al.} in~\cite{Tandon_synergistic} considered a MISO broadcast channel with alternating hybrid CSIT to be perfect, delayed, or no CSIT with respect to different receivers.


As mentioned earlier, investigation of broadcast channels under unequal link fading dynamics is fairly recent. An achievable degrees of freedom region for one static and one dynamic receiver was given in~\cite{Li_product,Li_coherent,Fadel_coherent,Fadel_coherence} via product superposition, producing a gain that is now known as {\em coherence diversity}. Coherence diversity gain was further investigated in~\cite{Fadel_disparity,Fadel_broadcast} for a $K$-receiver broadcast channel with neither CSIT nor CSIR. Also, a broadcast channel was investigated in \cite{Zhang_spatially}, where the receivers MIMO fading links experience {\em nonidentical} spacial correlation.


In this paper, we consider a multiuser model under a hybrid CSIR
scenario where a group of receivers, denoted static receivers, are
assumed to have CSIR, and another group with shorter link coherence
time, denoted dynamic receivers, do not have free CSIR. We consider
this model under a variety of CSIT conditions, including no CSIT, delayed CSIT, and two
hybrid CSIT scenarios. In each of these conditions, we analyze the
degrees of freedom region. A few new tools are introduced, and
inner and outer bounds are derived that partially meet in some
cases. The results of this paper are cataloged as follows.

In the absence of CSIT, an outer bound on the degrees of freedom region is produced via bounding the rates of a discrete memoryless multilevel broadcast channel~\cite{Borade_multilevel,Nair_capacity}
and then applying the  extremal entropy inequality~\cite{Liu_extremal,Liu_vector}. Our achievable degrees of freedom region meets the outer bound in the limiting case where the coherence times of the static and dynamic receivers are the same. 

For delayed CSIT, we use the outdated CSI model that was used by Maddah-Ali and Tse~\cite{Maddah_completely} under i.i.d.\ fading and assuming global CSIR at all nodes. 
Noting that our model does not have uniform CSIR, we produced a technique with alignment over super-symbols to utilize outdated CSIT but merge it together with product superposition to reuse the pilots of the dynamic receivers for the purpose of transmission to static receivers.
Moreover, we develop an outer bound that is suitable for block-fading channels with different coherence times, by appropriately enhancing the channel to a physically-degraded broadcast channel and then applying the extremal entropy inequality~\cite{Liu_extremal,Liu_vector}. For one static and one dynamic receiver, our achievable degrees of freedom partially meet our outer bound, and furthermore the gap decreases with the dynamic receiver coherence time $\T$. 

Under hybrid CSIT, we analyze two conditions: First, we consider perfect CSIT for the static receivers and no CSIT with respect to the dynamic receivers. The achievable degrees of freedom in this case are obtained using product superposition with the dynamic receiver's pilots reused and beamforming for the static receivers to avoid interference. Second, we consider perfect CSIT with respect to the static receivers and delayed CSIT with respect to the dynamic receivers. An achievable transmission scheme is proposed via a combination of beamforming, interference alignment, and product superposition methodologies. The outer bounds for the two hybrid-CSIT cases were based on constructing an enhanced physically degraded channel and then applying the extremal entropy inequality. For one static receiver with perfect CSIT and one dynamic receiver with delayed CSIT, the gap between the achievable and the outer sum degrees of freedom is $\frac{1}{\T}$.

\renewcommand{\arraystretch}{1.1}
\begin{table*}
\caption{Notation}
\begin{center}
\begin{tabular}{lcc}
 & {\bf Static Users} & {\bf Dynamic Users} \\
\hline
number of users & $\m'$ & $\m$ \\
MISO channel gains & $\g_1,\ldots,\g_{m'} $ & $\h_1,\ldots,\h_m $ \\
received signals (continuous) & $y'_1, \ldots, y'_{m'} $ & $y_1,\ldots,y_m $ \\
DMC receive variables & $Y'_1, \ldots, Y'_{m'} $ & $Y_1,\ldots,Y_m $ \\
transmission rates & $R'_1,\ldots, R'_{m'} $ & $R_1\ldots,R_m $ \\
messages & $M'_1,\ldots, M'_{m'} $ & $M_1,\ldots, M_m $ \\
degrees of freedom & $d'_1, \dots, d'_{m'} $ & $d_1, \dots, d_m $ \\
coherence time & $T' $ & $T $ \\
 \multicolumn{3}{c}{\bf General Variables} \\
\hline
$\X$ & transmit signal \\
$\rho$ & signal-to-noise ratio \\
$U_i, V_j, W$ & auxiliary random variables \\
$\Hset$ & set of all channel gains \\
${\Dset}_x$ & vertex of degrees of freedom region \\
${\e}_i$ & canonical coordinate vector \\
\hline
\end{tabular}
\end{center}
\end{table*}

\section{System Model}
\label{section:system}

\begin{figure}
\center
\includegraphics{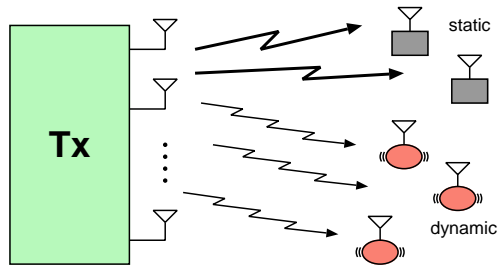}
\caption{A broadcast channel with multiple static and multiple dynamic users}
\label{figure:system_multiple_static_dynamic}
\end{figure}

Consider a broadcast channel with multiple single-antenna receivers and the transmitter is equipped with $\N_t$ antennas. The expressions ``receiver'' and ``user'' are employed without distinction throughout the paper, indicating the receiving terminals in the broadcast channel. The channels of the users are modeled as Rayleigh block fading where the channel coefficients remain constant over each block and change independently across blocks~\cite{Marzetta_capacity,Zheng_communication}. As shown in Fig.~\ref{figure:system_multiple_static_dynamic}, the users are partitioned into two sets based on channel availability and the length of the coherence interval: one set of dynamic users and another set of static users. The former contains $\m$ dynamic users having coherence time $\T$ and no free CSIR\footnote{This means that the cost of knowing CSI at the receiver, e.g., by channel estimation, is not ignored.}, and the latter contains $\m'$ static users having coherence time $\T'$ and perfect instantaneous CSIR. We consider the transmitter is equipped with more antennas than the number of dynamic and static users, i.e., $\N_t \geq \m' + \m$.

The received signals $y'_{j}(t), y_{i}(t)$ at the static user $j$, and the dynamic user $i$, respectively, at time instant $t$ are
\begin{align}
y'_{j}(t) & = \g_{j}^{\dagger}(t) \x(t) + z'_{j}(t), \quad j = 1, \ldots, \m', \nonumber \\
y_{i}(t)  & = \h_{i}^{\dagger}(t) \x(t) + z_{i}(t),  \quad i = 1, \ldots, \m, \label{eq:system_miso_bc}
\end{align}
where $\x(t) \in \mathbb{C}^{\N_t}$ is the transmitted signal, $z'_{j}(t), z_{i}(t)$ denote the corresponding additive i.i.d. Gaussian noise of the users, and $\g_{j}(t) \in \mathbb{C}^{\N_t}, \h_{i}(t) \in \mathbb{C}^{\N_t}$ denote the channels of the static user $j$ and the dynamic user $i$ whose coefficients stay the same over $\T'$ and $\T$ time instances, respectively. The distributions of $\g_j$ and $\h_i$ are globally known at the transmitter and at the users.\footnote{Also, the coherence times of all channels are globally known at the transmitter and at the users.} Having CSIR, the value of $\g_{j}(t)$ is available instantaneously and perfectly at the static user $j$. Furthermore, the static user $j$ obtains an outdated version of the dynamic users channels $\h_i$, and also the dynamic user $i$ obtains an outdated version of the static users channel $\g_i$ (completely stale)~\cite{Maddah_completely}. CSIT for each user can take one of the following:
\begin{itemize}
\item Perfect CSIT: the channel vectors, $\g_{j}(t),\h_{i}(t)$, are available at the transmitter instantaneously and perfectly.
\item Delayed CSIT: the channel vectors, $\g_{j}(t),\h_{i}(t)$, are available at the transmitter after they change independently in the following block (completely stale~\cite{Maddah_completely}).
\item No CSIT: the channel vectors, $\g_{j}(t),\h_{i}(t)$, cannot not be known at the transmitter.
\end{itemize}

We consider the broadcast channel with private messages for all users and no common messages. More specifically, we assume that the independent messages $\M'_j \in [1:2^{n \R'_i(\rho)}], \M_i \in [1:2^{n \R_i(\rho)}]$ associated with rates $\R'_{j}(\rho), \R_{i}(\rho)$ are communicated from the transmitter to the static user $j$ and dynamic user $i$, respectively, at $\rho$ signal-to-noise ratio. The degrees of freedom of the static and dynamic users achieving rates $\R'_{j}(\rho), \R_{i}(\rho)$ can be defined as
\begin{align}
d'_{j} & = \lim_{\rho \rightarrow \infty} \frac{R'_{j}(\rho)}{\log (\rho)}, \quad j= 1, \ldots, \m', \nonumber \\
d_{i} & = \lim_{\rho \rightarrow \infty} \frac{R_{i}(\rho)}{\log (\rho)}, \quad i = 1, \ldots, \m.
\end{align}
The degrees of freedom region is defined as
\begin{align}
\mathcal{D} = \big\{ & ( d'_1, \ldots, d'_{\m'}, d_1, \ldots, d_\m) \in \mathbb{R}_+^{\m' + \m} \big| \twocolbreak 
 \ \exists ( R'_1(\rho), \ldots, R'_{\m'}(\rho), R_1(\rho), \ldots, R_{\m}(\rho) ) \in C(\rho), \nonumber \\ 
 & d'_{j} = \lim_{\rho \rightarrow \infty} \frac{R'_{j}(\rho)}{\log (\rho)}, 
 d_{i} = \lim_{\rho \rightarrow \infty} \frac{R_{i}(\rho)}{\log (\rho)}, \twocolbreak
	 \ j = 1, \ldots, \m', \ i = 1, \ldots, \m \big\},
\end{align}
where $C(\rho)$ is the capacity region at $\rho$ signal-to-noise ratio. The sum degrees of freedom is defined as
\begin{equation}
d_{\text{sum}} = \lim_{\rho \rightarrow \infty} \frac{C_{\text{sum}}(\rho)}{\log(\rho)},
\end{equation}
where 
\begin{equation}
C_{\text{sum}}(\rho) = \max \sum_{j= 1}^{\m'} R'_{j}(\rho) + \sum_{i = 1}^{\m} R_{i}(\rho).
\end{equation}
In the sequel, we study the degrees of freedom of the above MISO broadcast channel under different CSIT scenarios that could be perfect, delayed or no CSIT.

\section{No CSIT for All Users}
\label{section:nn_CSIT}

In this section, we study the broadcast channel defined in Section~\ref{section:system} when there is no CSIT for all users. In particular, we give outer and achievable degrees of freedom regions in Section~\ref{section:nn_CSIT_outer} and Section~\ref{section:nn_CSIT_ach}, respectively. The outer degrees of freedom region is based on the construction of an outer bound on the rates of a multiuser multilevel discrete memoryless channel that is given in Section~\ref{section:Ml_BC}.

\subsection{Multiuser Multilevel Broadcast Channel}
\label{section:Ml_BC}

\begin{figure}
\center
\includegraphics[width=\Figwidth]{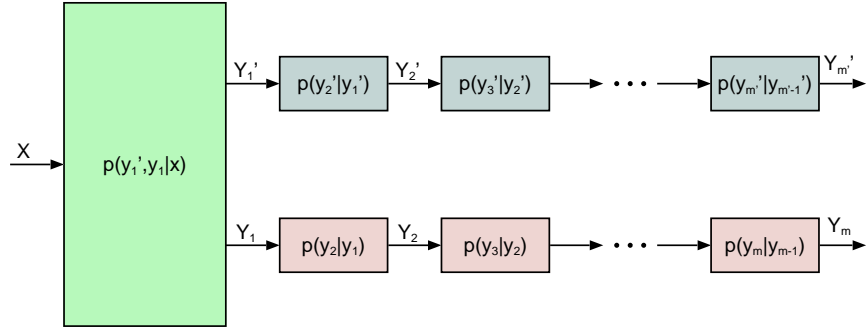}
\caption{Discrete memoryless multiuser multilevel broadcast channel}
\label{figure:Ml_BC}
\end{figure}

The multilevel broadcast channel was introduced by Borade {\em et al.}~\cite{Borade_multilevel} as a three-user broadcast discrete memoryless broadcast channel where two of the users are degraded with respect to each other. The capacity of this channel under degraded message sets was established by Nair and El Gamal~\cite{Nair_capacity}. Here, we study a multiuser multilevel broadcast channel with two sets of degraded users (see Fig.~\ref{figure:Ml_BC}). One set contains $\m'$ users with $Y'_j$ received signal at user $j$, and the other set contains $\m$ users with $Y_i$ received signal at user $i$. Therefore,
\begin{align}
X \rightarrow & Y'_{1} \rightarrow Y'_{2} \rightarrow \cdots \rightarrow Y'_{\m'} \nonumber \\
X \rightarrow & Y_{1} \rightarrow Y_{2} \rightarrow \cdots \rightarrow Y_{\m} 
\label{KM_outer_MC}
\end{align}
form two Markov chains. We consider a broadcast channel with $(\m'+\m)$ private messages and no common message. An outer bound for the above multilevel broadcast channel is given in the following theorem.
\begin{theorem}
\label{theorem:KM_outer}
The rate region of the multilevel broadcast channel with two sets of degraded users (Eq.~\eqref{KM_outer_MC}) is outer bounded by the intersection of
\begin{align}
\R_{1} \leq & I( U_{\m'}, W ; Y_1 | V_1 )
 - I( W ; Y'_{\m'} | U_{\m'} ), \label{eq:KM_outer_1} \\
\R_i \leq & I( V_{i-1} ; Y_i | V_i ), \qquad \qquad i=2, \ldots, \m, \label{eq:KM_outer_2} \\
\R'_j \leq & I( U_{j-1} ; Y'_j | U_j ), \qquad \qquad j=1, \ldots, \m'-1, \label{eq:KM_outer_3} \\
\R'_{\m'} \leq & I( W ; Y'_{\m'} | U_{\m'} ) 
 + I( X ; Y'_{\m'} | U_{\m'}, W ) \twocolbreak
- I( X ; Y'_{\m'} | U_{\m'-1} ), \label{eq:KM_outer_4} 
\end{align} 
and
\begin{align}
\R_i \leq & I( \tilde{U}_{i-1} ; Y_i | \tilde{U}_i ), \qquad \qquad i=1, \ldots, \m-1, \label{eq:KM_outer_5} \\
\R_{\m} \leq & I( \tilde{W} ; Y_{\m} | \tilde{U}_{\m} )
 + I( X ; Y_{\m} | \tilde{U}_{\m}, \tilde{W} ) - I( X ; Y_{\m} | \tilde{U}_{\m-1} ), \label{eq:KM_outer_6} \\
\R'_{1} \leq & I( \tilde{U}_{\m}, \tilde{W} ; Y'_1 | \tilde{V}_1 )
 - I( \tilde{W} ; Y_{\m} | \tilde{U}_{\m} ), \label{eq:KM_outer_7} \\
\R'_{j} \leq & I( \tilde{V}_{j-1} ; Y'_j | \tilde{V}_j ), \qquad \qquad j= 2, \ldots, \m', \label{eq:KM_outer_8}
\end{align}
for some pmf
\begin{equation}
p(u_1,\ldots,u_{m'}, \tilde{u}_1,\ldots,\tilde{u}_{m}, v_1,\ldots,v_m, \tilde{v}_1,\ldots,\tilde{v}_{m'}, w,\tilde{w}, x ),
\end{equation}
where
\begin{align}
U_{\m'} \rightarrow \cdots \rightarrow U_1 \rightarrow & X \rightarrow (Y_1, \ldots, Y_{\m}, Y'_1, \ldots, Y'_{\m'}) \nonumber \\
V_{\m} \rightarrow \cdots \rightarrow V_1 \rightarrow (W,U_{m'}) \rightarrow & X \rightarrow (Y_1, \ldots, Y_{\m}, Y'_1, \ldots Y'_{\m'}) \nonumber \\
\tilde{U}_{\m} \rightarrow \cdots \rightarrow \tilde{U}_1 \rightarrow & X \rightarrow (Y_1, \ldots Y_{\m}, Y'_1, \ldots, Y'_{\m'}) \nonumber \\
\tilde{V}_{\m'} \rightarrow \cdots \rightarrow \tilde{V}_1 \rightarrow (\tilde{W},\tilde{U}_{\m}) \rightarrow & X \rightarrow (Y_1, \ldots Y_{\m}, Y'_1, \ldots, Y'_{\m'})
\end{align}
forms Markov chains and $U_0=\tilde{U}_0\triangleq X$.
\end{theorem}
\begin{proof}
See Appendix~\ref{appendix:KM_outer}.
\end{proof}

\begin{remark}
\label{remark:KM_outer}
Theorem~\ref{theorem:KM_outer} is an extension of the K{\"o}rner-Marton outer bound~\cite[Theorem 5]{Marton_coding} to more than two users, and it recovers the K{\"o}rner-Marton bound when $\m=\m'=1$.
\end{remark}

\begin{remark}
For the multiuser multilevel broadcast channel characterized by~\eqref{KM_outer_MC}, we establish the capacity for degraded message sets in Appendix~\ref{appendix:degraded_message}, where one common message is communicated to all receivers and one further private message is communicated to one receiver.
\end{remark}

\subsection{Outer Degrees of Freedom Region}
\label{section:nn_CSIT_outer}

In the sequel, we give an outer bound on the degrees of freedom of the broadcast channel defined in Section~\ref{section:system} when there is no CSIT for all users. The outer bound development depends on the results of Theorem~\ref{theorem:KM_outer} in Section~\ref{section:Ml_BC}.
\begin{theorem}
\label{theorem:nn_CSIT_outer}
An outer bound on the degrees of freedom region of the fading broadcast channel characterized by Eq.~\eqref{eq:system_miso_bc} without CSIT is,
\begin{align}
 \sum_{j = 1}^{\m'} d'_{j} & \leq 1, \label{eq:nn_CSIT_outer_1} \\
 \sum_{i=1}^{\m} d_{i} & \leq 1 - \frac{1}{\T}, \label{eq:nn_CSIT_outer_2}\\
\sum_{j = 1}^{\m'} d'_{j} + \sum_{i = 1}^{\m} d_{i} & \leq 
\begin{cases}
1 & T=T', \ \Delta T = 0 \\
\frac{4}{3} & \text{otherwise},
\end{cases}
\label{eq:nn_CSIT_outer_3}
\end{align}
where $\Delta T$ is the offset between the two coherence intervals.
\end{theorem}

\begin{proof}
Equations~\eqref{eq:nn_CSIT_outer_1} and~\eqref{eq:nn_CSIT_outer_2} are outer bounds for a broadcast channel whose users are either all homogeneously static or all homogeneously dynamic~\cite{Fadel_disparity,Fadel_coherent}. The remainder of the proof is dedicated to establishing~\eqref{eq:nn_CSIT_outer_3}. We enhance the channel by giving all users global CSIR. When $\T'=\T$ and $\Delta T = 0$,~\eqref{eq:nn_CSIT_outer_3} follows directly from~\cite{Fadel_coherent,Fadel_disparity}. When $\T'\neq\T$ or $\Delta T \neq 0$, having no CSIT, the channel belongs to the class of multiuser multilevel broadcast channels in Section~\ref{section:Ml_BC}. We then use the two outer bounds developed for the multilevel broadcast channels to generate two degrees of freedom bounds, and merge them to get the desired result.

We begin with the outer bound described in~\eqref{eq:KM_outer_1}-\eqref{eq:KM_outer_4}; we combine these equations to obtain partial sum-rate bounds on the static $(\sum R'_j)$ and dynamic $(\sum R_i)$ receivers:
\begin{align}
\sum_{j = 1}^{\m'} \R'_{j}
\leq & \sum_{j=1}^{\m'-1} I( U_{j-1} ; y'_{j} | U_j, \Hset)
+ I( W ; y'_{\m'} | U_{\m'}, \Hset) \twocolbreak
+ I( \x ; y'_{\m'} | U_{\m'}, W, \Hset) \nonumber \\
& - I( \x ; y'_{\m'} | U_{\m'-1}, \Hset) \nonumber \\
= & \sum_{j=1}^{\m'-1} h( y'_{j} | U_j, \Hset)
- h( y'_{j} | U_{j-1},\Hset) \twocolbreak
+ I( W ; y'_{\m'} | U_{\m'}, \Hset) 
+ h( y'_{\m'} | U_{\m'}, W, \Hset) \nonumber \\ 
& - h( y'_{\m'} | U_{\m'-1}, \Hset) 
+ o(\log(\rho)) \label{eq:nn_CSIT_R_s_1} \\
= & I( W ; y'_{\m'} | U_{\m'}, \Hset)
+ h( y'_{\m'} | U_{\m'}, W, \Hset) \twocolbreak 
+ o(\log(\rho)), \label{eq:nn_CSIT_R_s_2}
\end{align}
where $\Hset$ is the set of all channel vectors,~\eqref{eq:nn_CSIT_R_s_1} follows from the chain rule, $h( y'_{j} | \x, \Hset) = o(\log(\rho))$, and~\eqref{eq:nn_CSIT_R_s_2} follows since the received signals of all static users, $y'_j$, have the same statistics~\cite{Fadel_coherent, Fadel_disparity}. Also, using Theorem~\ref{theorem:KM_outer},
\begin{align}
\sum_{j = 1}^{\m} \R_{j} 
\leq & I( U_{\m'}, W ; y_{1} | V_1, \Hset ) 
- I( W ; y'_{\m'} | U_{\m'}, \Hset ) \twocolbreak
 + \sum_{j=2}^{\m} I( V_{j-1}; y_{j} | V_j, \Hset ) \nonumber \\
= & h( y_{1} | V_1, \Hset ) 
- h( y_{1} | U_{\m'}, W, \Hset ) \twocolbreak 
- I( W ; y'_{\m'} | U_{\m'}, \Hset )
 + \sum_{j=2}^{\m} h( y_{j} | V_j, \Hset ) \nonumber \\
& - h( y_{j} | V_{j-1},\Hset ) \label{eq:nn_CSIT_R_d_1} \\
= & - h( y_{1} | U_{\m'}, W, \Hset ) 
- I( W ; y'_{\m'} | U_{\m'}, \Hset ) \twocolbreak
+ h( y_{\m}| V_{\m}, \Hset ) 
+ o(\log(\rho) ) \label{eq:nn_CSIT_R_d_2} \\
\leq & - h( y_{1} | U_{\m'}, W, \Hset ) 
- I( W ; y'_{\m'} | U_{\m'}, \Hset ) \twocolbreak
 + \log(\rho) 
 + o(\log(\rho) ), \label{eq:nn_CSIT_R_d_3}
\end{align}
where~\eqref{eq:nn_CSIT_R_d_1} follows from the chain rule,~\eqref{eq:nn_CSIT_R_d_2} follows since $ y_{j}$ have the same statistics, and~\eqref{eq:nn_CSIT_R_d_3} follows since $h( y_{\m} | V_{\m}, \Hset ) \leq n\log(\rho) + o(\log(\rho))$. Define $Y'_{j,k}$ to be the received signal of user $j$ at time instance $k$. From~\eqref{eq:nn_CSIT_R_s_2} and~\eqref{eq:nn_CSIT_R_d_3}, we can obtain the bound \eqref{eq:nn_CSIT_R_d_R_s_3} on the rates.
\makeatletter
\if@twocolumn%
\begin{figure*}
\begin{align}
\frac{1}{2} \sum_{j=1}^{\m'} \R'_{j} + \sum_{j=1}^{\m} \R_{j}
\leq & \frac{1}{2} I( W ; y'_{\m'} | U_{\m'}, \Hset )
 + \frac{1}{2} h( y'_{\m'}| U_{\m'}, W, \Hset ) 
 - h( y_{1} | U_{\m'}, W, \Hset ) 
 - I( W ; y'_{\m'} | U_{\m'}, \Hset )
 + \log(\rho) \nonumber\\
 & + o(\log(\rho)) \nonumber \\
= & \frac{1}{2} h( y'_{\m'}| U_{\m'}, W, \Hset ) 
 - h( y_{1} | U_{\m'}, W, \Hset )
 + \log(\rho) 
 + o(\log(\rho)) \nonumber \\
\leq & \frac{1}{2} h( y'_{\m'}, y_{1} | U_{\m'}, W, \Hset ) 
 - h( y_{1} | U_{\m'}, W, \Hset )
 + \log(\rho) 
 + o(\log(\rho)) \label{eq:nn_CSIT_R_d_R_s_1} \\
\leq & \sum_{k = 1}^n \frac{1}{2} h( y'_{\m',k}, y_{1,k} | U_{\m'}, W, \Hset, y'_{\m',1}, \ldots, y'_{\m',k-1} , 
y_{1,1}, \ldots, y_{1,k-1} ) \nonumber \\
& - h( y_{1,k} | U_{\m'}, W, \Hset, y'_{\m',1}, \ldots, y'_{\m',k-1}, y_{1,1}, \ldots, y_{1,k-1} ) + \log(\rho) + o(\log(\rho)) \label{eq:nn_CSIT_R_d_R_s_2} \\
\leq & \max_{\text{Tr} \{\Cov_x\} \leq \rho, \Cov_x \succcurlyeq 0} \mathbb{E}_H \big\{ \frac{1}{2} \log | \I + \H \Cov_x \H^{\dagger}| - \log (1 + \h_{1}^{\dagger} \Cov_x \h_{1} ) \big\} 
 + \log(\rho)
 + o(\log(\rho)). \label{eq:nn_CSIT_R_d_R_s_3}
\end{align}
\hrulefill
\end{figure*}
\else%
\begin{align}
\frac{1}{2} \sum_{j=1}^{\m'} \R'_{j} \twocolAlignMarker + \sum_{j=1}^{\m} \R_{j} \twocolnewline
\leq & \frac{1}{2} I( W ; y'_{\m'} | U_{\m'}, \Hset )
 + \frac{1}{2} h( y'_{\m'}| U_{\m'}, W, \Hset ) 
 - h( y_{1} | U_{\m'}, W, \Hset ) \nonumber \\
& - I( W ; y'_{\m'} | U_{\m'}, \Hset ) \twocolbreak
 + \log(\rho) 
 + o(\log(\rho)) \nonumber \\
= & \frac{1}{2} h( y'_{\m'}| U_{\m'}, W, \Hset ) 
 - h( y_{1} | U_{\m'}, W, \Hset ) \twocolbreak
 + \log(\rho) 
 + o(\log(\rho)) \nonumber \\
\leq & \frac{1}{2} h( y'_{\m'}, y_{1} | U_{\m'}, W, \Hset ) 
 - h( y_{1} | U_{\m'}, W, \Hset ) \twocolbreak
 + \log(\rho) 
 + o(\log(\rho)) \label{eq:nn_CSIT_R_d_R_s_1} \\
\leq & \sum_{k = 1}^n \frac{1}{2} h( y'_{\m',k}, y_{1,k} | U_{\m'}, W, \Hset, y'_{\m',1}, \ldots, y'_{\m',k-1}, 
 y_{1,1}, \ldots, y_{1,k-1} ) \nonumber \\
& - h( y_{1,k} | U_{\m'}, W, \Hset, 
 y'_{\m',1}, \ldots, y'_{\m',k-1}, 
 y_{1,1}, \ldots, y_{1,k-1} ) + \log(\rho) \nonumber \\
& + o(\log(\rho)) \label{eq:nn_CSIT_R_d_R_s_2} \\
\leq & \max_{\text{Tr} \{\Cov_x\} \leq \rho, \Cov_x \succcurlyeq 0} \mathbb{E}_H \big\{ \frac{1}{2} \log | \I + \H \Cov_x \H^{\dagger}| \twocolbreak
 - \log (1 + \h_{1}^{\dagger} \Cov_x \h_{1} ) \big\} 
 + \log(\rho) \nonumber \\
 & + o(\log(\rho)), \label{eq:nn_CSIT_R_d_R_s_3}
\end{align}
\fi%
\makeatother
where~\eqref{eq:nn_CSIT_R_d_R_s_1} and~\eqref{eq:nn_CSIT_R_d_R_s_2} follow from the chain rule and that conditioning does not increase differential entropy, and~\eqref{eq:nn_CSIT_R_d_R_s_3} follows from extremal entropy inequality~\cite{Liu_extremal,Liu_vector,Yang_degrees_1}. In order to bound~\eqref{eq:nn_CSIT_R_d_R_s_3}, we use a specialization of~\cite[Lemma 3]{Yi_degrees} as follows.
\begin{lemma}
\label{lemma:DoF_bounding}
Consider two random matrices $\H_1 \in \mathbb{C}^{ N_1 \times \N_t}$ and $\H_2 \in \mathbb{C}^{N_2 \times \N_t}$, where $N_1 \geq N_2$. For a covariance matrix $\Cov_x$, where $\text{Tr}\{ \Cov_x \} \leq \rho$, we have 
\begin{align}
\max_{\Cov_x} \ & \frac{1}{\min\{ \N_t, N_1\}} \log | \I + \H_1 \Cov_x \H_1^{\dagger}| \twocolbreak
- \frac{1}{\min\{ \N_t, N_2\}} \log | \I + \H_2 \Cov_x \H_2^{\dagger}| \twocolbreak
 \leq o(\log(\rho)).
\end{align}
\end{lemma}
The proof of Lemma~\ref{lemma:DoF_bounding} is omitted as it directly follows from~\cite[Lemma 3]{Yi_degrees}. Lemma~\ref{lemma:DoF_bounding} yields the following outer bound on the degrees of freedom:
\begin{equation}
\frac{1}{2} \sum_{j =1}^{\m'} d'_{j} + \sum_{i =1}^{\m} d_{i} \leq 1.
\label{eq:nn_CSIT_outer_11}
\end{equation}
We now repeat the exercise of bounding the sum rates and deriving degrees of freedom, this time starting from~\eqref{eq:KM_outer_5}-\eqref{eq:KM_outer_8}. By following bounding steps parallel to~\eqref{eq:nn_CSIT_R_s_2},~\eqref{eq:nn_CSIT_R_d_3},~\eqref{eq:nn_CSIT_R_d_R_s_3},
\begin{equation}
\sum_{j =1}^{\m'} d'_{j} + \frac{1}{2} \sum_{i =1}^{\m} d_{i} \leq 1.
\label{eq:nn_CSIT_outer_12}
\end{equation}
Adding~\eqref{eq:nn_CSIT_outer_11} and~\eqref{eq:nn_CSIT_outer_12} yields the outer bound~\eqref{eq:nn_CSIT_outer_3}, completing the proof of Theorem~\ref{theorem:nn_CSIT_outer}.
\end{proof}

\subsection{Achievable Degrees of Freedom Region}
\label{section:nn_CSIT_ach}

\begin{theorem}
\label{theorem:nn_CSIT_ach}
The fading broadcast channel described by Eq.~\eqref{eq:system_miso_bc} can achieve the following degrees of freedom without CSIT:
\begin{align}
\sum_{i =1}^{\m} d_{i} & \leq 1 - \frac{1}{\T}, \label{eq:nn_CSIT_ach_1} \\
\sum_{j =1}^{\m'} d'_{j} + \sum_{i =1 }^{\m} d_{i} & \leq 1. \label{eq:nn_CSIT_ach_2}
\end{align}
\end{theorem}

\begin{proof}
The achievable scheme uses product superposition~\cite{Li_coherent,Fadel_coherent}, where the transmitter uses one antenna to send the super symbol to two users: one dynamic and one static,
\begin{equation}
\x^{\dagger}= x_s \x_d^{\dagger},
\end{equation}
where $x_s \in \mathbb{C}$ is a symbol intended for the static user,
\begin{equation}
\x_d^{\dagger} = [ x_{\tau},\; \; \x_{\delta}^{\dagger} ]
\end{equation}
where $ x_{\tau} \in \mathbb{C}$ is a pilot and $\x_{\delta} \in \mathbb{C}^{ \T - 1 }$ is a super symbol intended for the dynamic user. Since degrees of freedom analysis is insensitive to the additive noise, we omit the noise component in the following.
\begin{align}
\y^{\dagger} &= h x_s [ x_{\tau}, \;\; \x_{\delta}^{\dagger} ] \nonumber \\ 
&= [ \overline{h} x_{\tau}, \;\; \overline{h} \x_{\delta}^{\dagger} ],
\end{align}
where $\overline{h}=h x_s$. The dynamic user estimates the equivalent channel $\overline{h}$ during the first time instance and then decodes $\x_{\delta}$ {\em coherently} based on the channel estimate. The static receiver only utilizes the received signal during the first time instance:
\begin{equation}
y'_{1} = g x_s.
\end{equation}
Knowing its channel gain $g$, the static receiver can decode $x_s$. The achievable degrees of freedom of the two users are,
\begin{equation}
(d',d) = \big( \frac{1}{\T}, 1 - \frac{1}{\T} \big).
\end{equation}
We now proceed to prove that the degrees of freedom region characterized by~\eqref{eq:nn_CSIT_ach_1} and~\eqref{eq:nn_CSIT_ach_2} can be achieved via a combination of  two-user product superposition strategies that were outlined above, and single-user strategies. For clarity of exposition we refer to~\eqref{eq:nn_CSIT_ach_1}, which describes the degrees of freedom constraints of the dynamic receivers, as the {\em non-coherent bound}, and to~\eqref{eq:nn_CSIT_ach_2} as the {\em coherent bound}. The non-negativity of degrees of freedom restricts them to the non-negative orthant ${\mathbb R}_+^{\m+\m'}$. The intersection of the coherent bound and the non-negative orthant is a $(\m'+\m)$--simplex that has $\m+\m'+1$ vertices. The non-coherent bound is a hyperplane that partitions the simplex with $\m'+1$ vertices on one side of the non-coherent bound and $\m$ on the other. Therefore the intersection of the simplex with the non-coherent bound produces a polytope with $(\m'+1)(\m+1)$ vertices.\footnote{This can be verified with a simple counting exercise involving the number of edges of the simplex that cross the non-coherent bound.} For illustration, see Fig.~\ref{figure:region_2static_1dynamic} showing the three-user degrees of freedom with two static users and Fig.~\ref{figure:region_1static_2dynamic} with one static user. 

\begin{figure}
\begin{minipage}[b]{3.in}
\center
\includegraphics[width=3.5in]{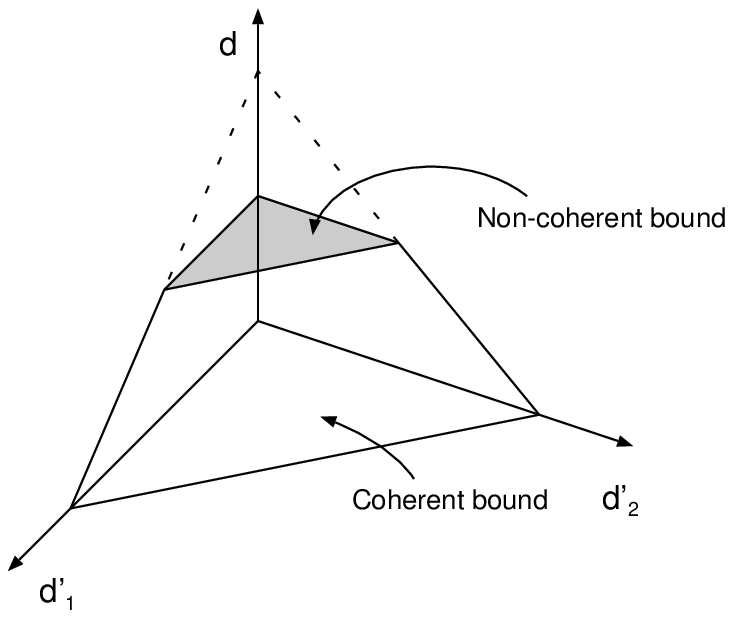}
\caption{Achievable degrees of freedom region of one dynamic and two static users}
\label{figure:region_2static_1dynamic}
\end{minipage}
\hfill
\begin{minipage}[b]{3.in}
\center
\includegraphics[width=3.5in]{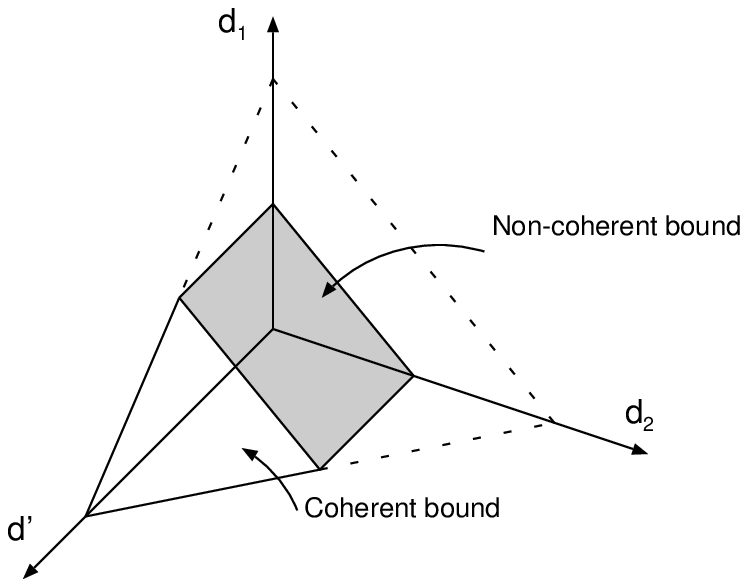}
\caption{Achievable degrees of freedom region of one static and two dynamic users}
\label{figure:region_1static_2dynamic}
\end{minipage}
\end{figure}

We now verify that each of the $(\m'+1)(\m+1)$ vertices can be achieved with either a single-user strategy, or via a two-user product superposition strategy:
\begin{itemize}
\item $\m'$ vertices corresponding to single-user transmission to each static user $j$ achieving one degree of freedom.
\item $\m$ vertices corresponding to single-user transmission to each dynamic user $i$ achieving $(1 - \frac{1}{ \T})$ degrees of freedom.
\item $\m' \m$ vertices corresponding to product superposition applied to all possible pairs of static and dynamic users, achieving $\frac{1}{\T}$ degrees of freedom for one static user and $( 1 - \frac{1}{\T} )$ degrees of freedom for one dynamic user.
\item One trivial vertex at the origin, corresponding to no transmission achieving zero degrees of freedom for all users.
\end{itemize}
Hence, the number of the vertices is $\m' + \m + \m'\m +1 = (\m+1)(\m'+1)$. This completes the achievability proof of Theorem~\ref{theorem:nn_CSIT_ach}.
\end{proof}

\begin{remark}
\label{remark:nn_CSIT_identical}
When the static and dynamic users have the same coherence time, the inner and outer bounds on degrees of freedom coincide. In this case it is degrees of freedom optimal to serve two-users at a time (one dynamic and one static).
\end{remark}

\section{Delayed CSIT for all users}
\label{section:dd_CSIT}
Under delayed CSIT, the transmitter knows each channel gain only after it is no longer valid. This condition is also known as outdated CSIT. We begin by proving inner and outer bounds when transmitting only to static users, only to dynamic users, and to one static and one dynamic user. We then synthesize this collection of bounds into an overall degrees of freedom region. 

\subsection{Transmission to Static Users}
\label{section:dd_CSIT_static}

\begin{theorem}
\label{theorem:dd_CSIT_static}
The degrees of freedom region of the fading broadcast channel characterized by Eq.~\eqref{eq:system_miso_bc} with delayed CSIT and having $\m'$ static users and no dynamic users is
\begin{equation}
d'_{j} \leq \frac{1}{ 1 + \frac{1}{2} + \ldots + \frac{1}{\m'}}, \quad j = 1, \ldots, \m'.
\end{equation}
\end{theorem}

\begin{proof}
The special case of fast fading ($\T'=1$) was discussed by Maddah-Ali and Tse in~\cite{Maddah_completely}, where the achievability was established by {\em retrospective interference alignment} that aligns the interference using the outdated CSIT, and the converse was proved by generating an improved channel without CSIT having a tight degrees of freedom region against TDMA according to the results in~\cite{Huang_degrees,Vaze_degree}. For $\T' \geq 1$, the achievability is established by employing retrospective interference alignment presented in~\cite{Maddah_completely} over super symbols each of length $\T'$. The converse is proved by following the same procedures in~\cite{Maddah_completely} to generate a block-fading improved channel without CSIT and with identical coherence intervals of length $\T'$. According to the results of~\cite{Fadel_coherent,Fadel_disparity}, TDMA is tight against the degrees of freedom region of the improved channel. 
\end{proof}

\subsection{Transmission to Dynamic Users}
\label{section:dd_CSIT_dynamic}

\begin{theorem}
\label{theorem:dd_CSIT_dynamic}
The fading broadcast channel characterized by Eq.~\eqref{eq:system_miso_bc}, with delayed CSIT and having $\m$ dynamic users and no static users, can achieve the degrees of freedom
\begin{equation}
\label{eq:dd_CSIT_dynamic_ach}
d_{i} \leq \frac{1}{ 1 + \frac{1}{2} + \ldots + \frac{1}{\m}} ( 1 - \frac{\m}{\T}), \quad i = 1, \ldots, \m.
\end{equation}
An outer bound on the degrees of freedom region is
\begin{align}
d_{i} \leq & 1 - \frac{1}{\T}, \label{eq:dd_CSIT_dynamic_outer_1} \\
\sum_{i = 1}^{\m} d_{i} \leq & \frac{\m}{ 1 + \frac{1}{2} + \ldots + \frac{1}{\m}}. \label{eq:dd_CSIT_dynamic_outer_2}
\end{align}
\end{theorem}

\begin{proof}
The achievability part can be proved as follows. At the beginning of each super symbol, $\m$ pilots are sent for channel estimation. Then retrospective interference alignment in~\cite{Maddah_completely} over super symbols is employed during the remaining $(\T - \m)$ instances, to achieve~\eqref{eq:dd_CSIT_dynamic_ach}. For the converse part,~\eqref{eq:dd_CSIT_dynamic_outer_2} is proved by giving the users global CSIR, and then applying Theorem~\ref{theorem:dd_CSIT_static}. Moreover,~\eqref{eq:dd_CSIT_dynamic_outer_1} is the single-user bound for each dynamic user that can be proved as follows. For a single user with delayed CSIT, feedback does not increase the capacity~\cite{Shannon_zero}, and consequently the assumption of delayed CSIT can be removed. Hence, the single-user bound for each dynamic user with delayed CSIT is the same as the single-user bound without CSIT~\cite{Zheng_communication}. 
\end{proof}

\subsection{Transmission to One Static and One Dynamic User}
\label{section:dd_CSIT_static_dynamic_one}

\begin{theorem}
\label{theorem:dd_CSIT_static_dynamic_one_ach}
The fading broadcast channel characterized by Eq.~\eqref{eq:system_miso_bc}, with delayed CSIT and having one static and one dynamic user, can achieve the following degrees of freedom
\begin{align}
\Dset_1 &: (d',d) = \big(\frac{2}{3} ( 1 + \frac{1}{\T} ), \frac{2}{3} ( 1 - \frac{2}{\T} ) \big) \label{eq:dd_csit_ach_1}, \\
\Dset_2 &: (d',d) = (\frac{1}{\T}, 1 - \frac{1}{\T} ) \label{eq:dd_csit_ach_2}.
\end{align}
Furthermore, the achievable degrees of freedom region is the convex hull of the above degrees of freedom pairs.
\end{theorem}

\begin{proof}
From Section~\ref{section:nn_CSIT_ach}, product superposition achieves the pair~\eqref{eq:dd_csit_ach_2} that does not require CSIT for any of the two users. The remainder of the proof is dedicated to the achievability of the pair~\eqref{eq:dd_csit_ach_1}. We provide a transmission scheme based on retrospective interference alignment~\cite{Maddah_completely} along with product superposition.
\begin{enumerate}
\item The transmitter first emits a super-symbol intended for the static user:
\begin{equation}
\X_1 = [ \X_{1,1},\;\; \cdots, \;\; \X_{1,\ell} ],
\end{equation}
where $\ell = \frac{\T'}{\T}$, and each $\X_{1,n} \in \mathbb{C}^{2 \times \T}$ occupies $T$ time instances and has the following structure:
\begin{equation}
\X_{1,n} = [ \bar{\U}_n, \;\; \bar{\U}_n \U_n ], \quad n = 1, \ldots, \ell,
\end{equation}
both the diagonal matrix $\bar{\U}_n \in \mathbb{C}^{2 \times 2}$ and $\U_n \in \mathbb{C}^{2 \times (\T-2)}$ contain symbols intended for the static user. The components of  $\y_1^{'\dagger} = [ \y_{1,1}^{'\dagger}, \;\; \cdots, \;\; \y_{1,\ell}^{'\dagger} ]$ are:
\begin{align}
\y_{1,n}^{'\dagger} = & [ \g_1^\dagger \bar{\U}_n, \;\; \g_1^\dagger \bar{\U}_n \U_n ], \quad n = 1, \ldots, \ell \nonumber\\
= & [ \tilde{\g}_{1,n}^\dagger, \;\; \tilde{\g}_{1,n}^\dagger \U_n ],
\end{align}
where $\tilde{\g}_{1,n}^\dagger = \g_1^\dagger \bar{\U}_n$. The static user by definition knows $\g_1$ so it can decode $\bar{\U}_n$ which yields $2\frac{\T'}{\T}$ degrees of freedom. The remaining $\frac{\T'}{\T}\left(\T - 2\right)$ observations in $\tilde{\g}_{1,n}^\dagger\U_n$ involve $2\frac{\T'}{\T}\left(\T - 2\right)$ unknowns, so they require a further $\frac{\T'}{\T}\left(\T - 2\right)$ independent observations for reliable decoding. 

The components of $\y_1^{\dagger} = [ \y_{1,1}^{\dagger}, \;\; \cdots, \;\; \y_{1,\ell}^{\dagger}]$ are
\begin{align}
\y_{1,n}^{\dagger} = & [ \h_{1,n}^\dagger \bar{\U}_n, \;\; \h_{1,n}^\dagger \bar{\U}_n \U_n ], \quad n = 1, \ldots, \ell \nonumber\\
 = & [ \tilde{\h}_{1,n}^\dagger, \;\; \tilde{\h}_{1,n}^\dagger \U_n ],
\end{align}
where $\tilde{\h}_{1,n}^\dagger = \h_{1,n}^\dagger \bar{\U}_n$ is the equivalent channel estimated by the dynamic user. The dynamic user saves $\tilde{\h}_{1,n}^\dagger \U_n$ for interference cancellation in the upcoming steps.

\item The transmitter sends a second super symbol intended for the dynamic user:
\begin{equation}
\X_2 = [ \X_{2,1}, \;\; \cdots, \;\; \X_{2,\ell} ],
\end{equation}
where
\begin{equation}
\X_{2,n} = [ \tilde{\U}_n, \;\; \tilde{\U}_n \V_n ], \quad n = 1, \ldots, \ell,
\end{equation}
$\tilde{\U}_n \in \mathbb{C}^{2 \times 2}$ is diagonal and includes 2 independent symbols intended for the static user, and $\V_n \in \mathbb{C}^{2 \times (\T -2)}$ contains independent symbols intended for the dynamic user. The components of
$\y_2^{\dagger} = [ \y_{2,1}^{\dagger}, \;\; \cdots, \;\; \y_{2,\ell}^{\dagger} ]$
are
\begin{align}
\y_{2,n}^{\dagger} = & [ \h_{2,n}^\dagger \tilde{\U}_n, \;\; \h_{2,n}^\dagger \tilde{\U}_n \V_n ], \quad n = 1, \ldots, \ell \nonumber\\
= & [ \tilde{\h}_{2,n}^\dagger, \;\; \tilde{\h}_{2,n}^\dagger\V_n ],
\end{align}
where $\tilde{\h}_{2,n}^\dagger = \h_{2,n}^\dagger \tilde{\U}_n$ is the equivalent channel estimated by the dynamic user. The dynamic user saves $\tilde{\h}_{2,n}^\dagger\V_n$ which includes $\frac{\T'}{\T}\left(\T - 2\right)$ independent observations about $2\frac{\T'}{\T}\left(\T - 2\right)$ unknowns, and hence an additional $\frac{\T'}{\T}\left(\T - 2\right)$ observations are needed to decode $\V_n$. The components of $\y_2^{'\dagger} = [ \y_{2,1}^{'\dagger}, \;\; \cdots, \;\; \y_{2,\ell}^{'\dagger} ]$ are
\begin{align}
\y'_{2,n} = & [ \g_2^\dagger \tilde{\U}_n, \;\; \g_2^\dagger \tilde{\U}_n \V_n ], \quad n = 1, \ldots, \ell \nonumber\\
= & [ \tilde{\g}_{2,n}^\dagger, \;\; \tilde{\g}_{2,n}^\dagger\V_n ],
\end{align}
where $\tilde{\g}_{2,n}^\dagger = \g_2^\dagger \tilde{\U}_n$ is the equivalent channel estimated by the static user; the static user saves $\tilde{\g}_{2,n}^\dagger \V_n$ for the upcoming steps. Knowing $\g_2$, the static user achieves $2\frac{\T'}{\T}$ further degrees of freedom from decoding $\tilde{\U}_n$.

\item The transmitter emits a third super symbol consisting of a linear combination of the signals generated from the first and the second super symbols.
\begin{equation}
\X_3 = [ \X_{3,1}, \;\; \cdots, \;\; \X_{3,\ell} ], 
\end{equation}
where
\begin{equation}
\X_{3,n} = [ \hat{\U}_n, \;\; \hat{\U}_n ( \tilde{\h}_{1,n}^\dagger\U_n + \tilde{\g}_{2,n}^\dagger\V_n ) ], \quad n = 1, \ldots, \ell,
\end{equation}
$\hat{\U}_n \in \mathbb{C}^{2 \times 2}$ is diagonal and contains 2 independent symbols intended for the static user, and hence the static user achieves further $2\frac{\T'}{\T}$ degrees of freedom.

The static user cancels $\tilde{\g}_{2,n}^\dagger\V_n$ saved during the second super symbol and obtains $\tilde{\h}_{1,n}^\dagger\U_n$ that includes the additional independent $\frac{\T'}{\T}\left( \T - 2\right)$ observations needed for decoding $\U_n$. Therefore, the static user achieves $2\frac{\T'}{\T}( \T - 2)$ further degrees of freedom. The dynamic user estimates the equivalent channel $\tilde{\h}_{3,n}^\dagger = \h_{3,n}^\dagger \hat{\U}_n$, cancels $\tilde{\h}_{1,n}^\dagger\U_n$ saved during the first super symbol, and obtains $\tilde{\g}_{2,n}^\dagger\V_n$ that contains the additional observations needed for decoding $\V_n$. Hence, the dynamic user achieves $2\frac{\T'}{\T} ( \T - 2 )$ degrees of freedom.
\end{enumerate}
In aggregate, over $3\T'$ time instants, the static and dynamic user achieve the degrees of freedom
\begin{equation}
d'=6\frac{\T'}{\T}+ 2\frac{\T'}{\T} ( \T - 2), \qquad d= 2\frac{\T'}{\T} ( \T - 2 ).
\end{equation}
This completes the proof of Theorem~\ref{theorem:dd_CSIT_static_dynamic_one_ach}.
\end{proof}

\begin{theorem}
\label{theorem:dd_CSIT_static_dynamic_one_outer}
An outer bound on the degrees of freedom region of the fading broadcast channel characterized by Eq.~\eqref{eq:system_miso_bc}, with one static and one dynamic user having delayed CSIT, is
\begin{align}
\frac{d'}{2} + d &\leq 1, \label{eq:dd_CSIT_static_dynamic_one_outer_1} \\
d' + \frac{d}{2} &\leq 1, \label{eq:dd_CSIT_static_dynamic_one_outer_2} \\
 d &\leq 1 - \frac{1}{\T}. \label{eq:dd_CSIT_static_dynamic_one_outer_3}
\end{align}
\end{theorem}

\begin{proof}
The inequality~\eqref{eq:dd_CSIT_static_dynamic_one_outer_3} represents the single-user outer bound~\cite{Zheng_communication}. We prove the bound~\eqref{eq:dd_CSIT_static_dynamic_one_outer_1} as follows. We enhance the original channel by giving both users global CSIR. In addition, the channel output of the dynamic user, $y(t)$, is given to the static user. Therefore, the channel outputs at time instant $t$ are $( y'(t), y(t), \Hset )$ at the static user, and $( y(t), \Hset )$ at the dynamic user. The enhanced channel is physically degraded~\cite{Bergmans_random, Bergmans_simple}, hence, removing the delayed CSIT does not reduce the capacity~\cite{Gamal_feedback}. Also,
\begin{align}
R' \leq & I ( x(t) ; y'(t), y(t) | U,\Hset ) \twocolbreak
  = h ( y'(t), y(t) | U,\Hset ) 
 - h ( y'(t), y(t) | U, x(t), \Hset ) \nonumber \\
R \leq & I ( U ; y(t) | \Hset ) 
= h ( y(t) | \Hset ) 
 - h ( y(t) | U,\Hset ),
\end{align}
where $U$ is an auxiliary random variable, and $U \rightarrow x \rightarrow (y'(t), y(t))$ forms a Markov chain. Therefore, 
\begin{align}
\frac{R'}{2} + R \leq & h ( y(t) | \Hset) 
 + \frac{1}{2} h ( y'(t), y(t) | U, \Hset) \twocolbreak
 - h ( y(t) | U, \Hset) 
 + o ( \log(\rho) ) \nonumber \\ 
\leq & \log(\rho)
+ \frac{1}{2} h ( y'(t), y(t) | U, \Hset) \twocolbreak
- h ( y(t) | U, \Hset) 
 + o( \log(\rho) ) \label{eq:dd_CSIT_outer_Rs_Rd_1} \\
\leq &\log(\rho)
+ \max_{\text{Tr} \{\Cov_x\} \leq \rho, \Cov_x \succcurlyeq 0} 
 \mathbb{E}_H \big\{ \frac{1}{2} \log | \I + \H \Cov_x \H^\dagger | \twocolbreak
 - \log(1 + \h^\dagger(t) \Cov_x \h(t) ) \big\} 
+ o(\log(\rho) ) \label{eq:dd_CSIT_outer_Rs_Rd_2} \\
 \leq & \log(\rho) + o(\log(\rho) ), \label{eq:dd_CSIT_outer_Rs_Rd_3}
\end{align}
where~\eqref{eq:dd_CSIT_outer_Rs_Rd_1} follows since $h(y(t)| \Hset) \leq \log(\rho) + o(\log(\rho))$~\cite{Telatar_capacity},~\eqref{eq:dd_CSIT_outer_Rs_Rd_2} follows from extremal entropy inequality~\cite{Yi_degrees,Liu_extremal,Liu_vector}, and~\eqref{eq:dd_CSIT_outer_Rs_Rd_3} follows from Lemma~\ref{lemma:DoF_bounding}. Hence, the bound~\eqref{eq:dd_CSIT_static_dynamic_one_outer_1} is proved.
A similar argument, with the role of the two users reversed, leads to the bound~\eqref{eq:dd_CSIT_static_dynamic_one_outer_2}.
\end{proof}

\begin{remark}
\label{remark:delayed_static_dynamic_one_opt}
The inner and outer bounds obtained for the two-user case partially
meet, with the gap diminishing with
the coherence time of the dynamic user as shown in
Fig.~\ref{figure:delayed_static_dynamic_one_T15} and
Fig.~\ref{figure:delayed_static_dynamic_one_T30} for $\T = 15$ and $\T
= 30$, respectively.
\end{remark}

\begin{figure}
\center
\includegraphics[width=\Figwidth]{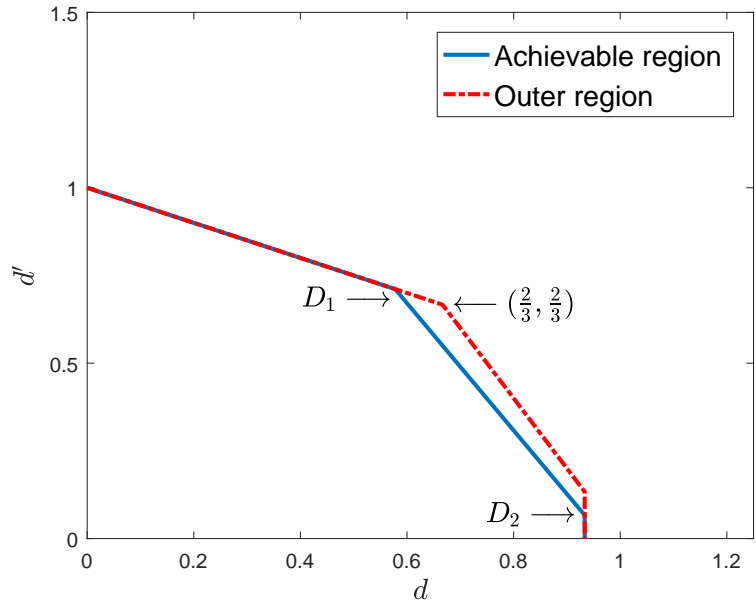}
\caption{One static and one dynamic with delayed CSIT and $\T = 15$}
\label{figure:delayed_static_dynamic_one_T15}
\end{figure}
\begin{figure}
\center
\includegraphics[width=\Figwidth]{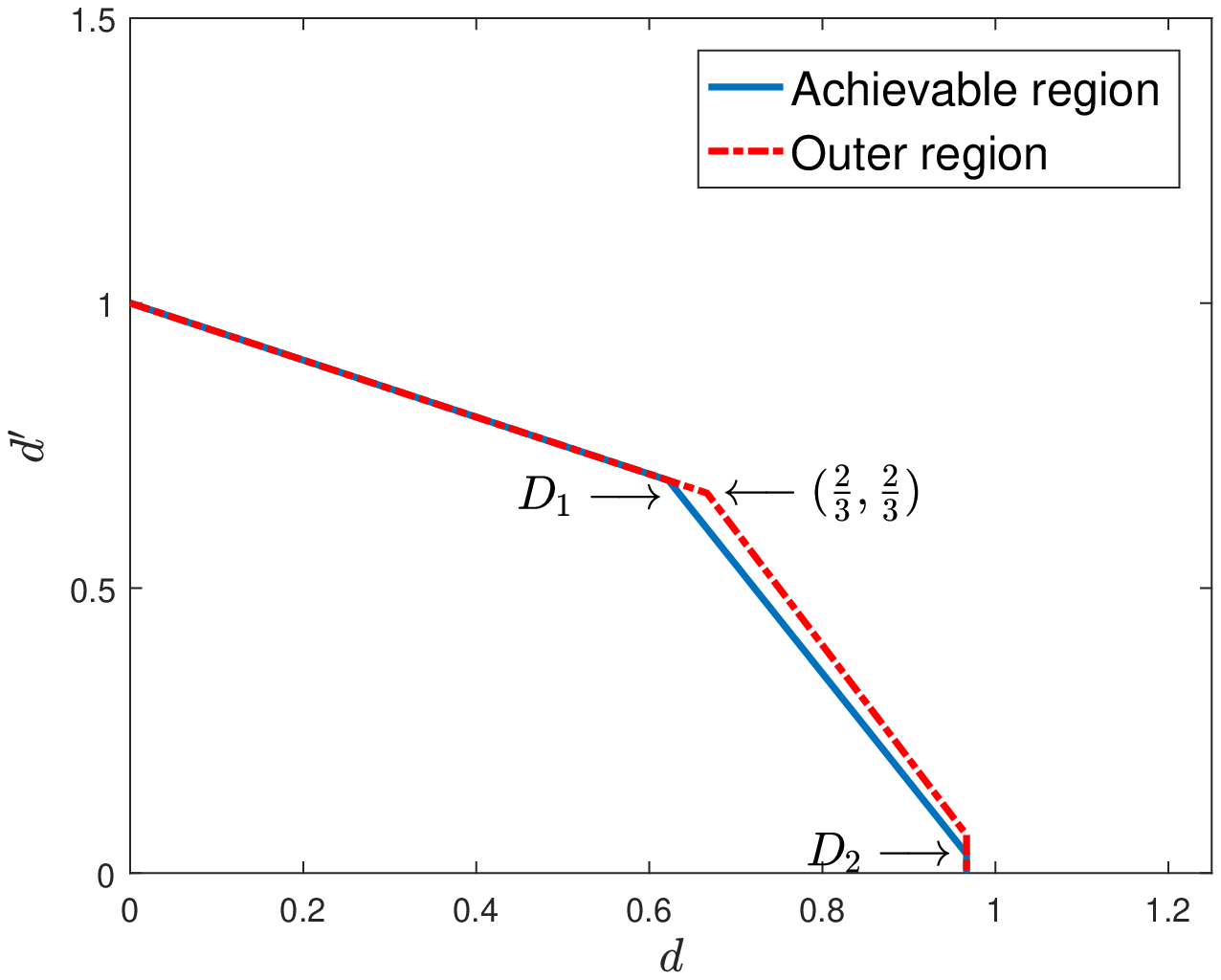}
\caption{One static and one dynamic with delayed CSIT and $\T = 30$}
\label{figure:delayed_static_dynamic_one_T30}
\end{figure}

\subsection{Transmission to arbitrary number of static and dynamic users}
\label{section:dd_CSIT_static_dynamic_K}

\begin{theorem}
\label{theorem:dd_CSIT_static_dynamic_K}
The fading broadcast channel characterized by Eq.~\eqref{eq:system_miso_bc}, with delayed CSIT, can achieve the multiuser degrees of freedom characterized by vectors $\Dset_i$,
\begin{align}
\Dset_1 & : 
\frac{1}{1 +\frac{1}{2}+\ldots + \frac{1}{\m'}} \sum_{i=1}^{\m'} \e_i^{\dagger}, \label{eq:dd_CSIT_ach_Kusers_1} \\
\Dset_2, & \ldots, \Dset_{\m\m'+1}: 
\frac{2}{3}( 1 + \frac{1}{\T}) \e_j^{\dagger} + \frac{2}{3}( 1 - \frac{2}{\T}) \e_{\m'+i}^{\dagger}, \twocolbreak \ j = 1, \ldots, \m' , \quad i = 1, \ldots, \m,
 \label{eq:dd_CSIT_ach_Kusers_2} \\
\Dset_{\m\m'+2}, & \ldots, \Dset_{\m\m'+\m'+2} :
\frac{\m}{\T} \e_j^{\dagger} + \twocolbreak \frac{1}{1 + \frac{1}{2} + \ldots + \frac{1}{\m}} (1 - \frac{\m}{\T} ) \sum_{i=1}^{\m} \e_i^{\dagger}, 
\quad j = 1, \ldots, \m',
 \label{eq:dd_CSIT_ach_Kusers_3}
\end{align}
where $\e_j$ is the canonical coordinate vector. Their convex hull characterized an achievable degrees of freedom region.
\end{theorem}

\begin{proof}
The achievability of~\eqref{eq:dd_CSIT_ach_Kusers_1} was proved in Section~\ref{section:dd_CSIT_static} via multiuser transmission to static users. The achievability of~\eqref{eq:dd_CSIT_ach_Kusers_2} was proved in Section~\ref{section:dd_CSIT_static_dynamic_one}, via a two-user transmission to a dynamic-static pair. 

We now show the achievability of~\eqref{eq:dd_CSIT_ach_Kusers_3} via retrospective interference alignment~\cite{Maddah_completely} along with product superposition. Over a super symbol of length $\T$, consider the following transmission:
\begin{equation}
\X = [ \U, \;\; \U \V ],
\end{equation}
where $\U \in \mathbb{C}^{\m \times \m}$ is diagonal and includes $\m$ independent symbols intended for the static user $j$, and $\V \in \mathbb{C}^{\m \times (\T - \m)}$ is a super symbol containing independent symbols intended for the dynamic users according to retrospective interference alignment~\cite{Maddah_completely}. Therefore, the static user decodes $\U$. Thus, over $\T$ time instants, the static user achieves $\m$ degrees of freedom and the dynamic users achieve $\frac{1}{1 + \frac{1}{2} + \ldots + \frac{1}{\m}} ( \T - \m )$, hence~\eqref{eq:dd_CSIT_ach_Kusers_3} is achieved.
\end{proof}

\begin{theorem}
\label{theorem:dd_CSIT_outer_Kusers}
An outer bound on the degrees of freedom of the fading broadcast channel characterized by Eq.~\eqref{eq:system_miso_bc}, with delayed CSIT, is
\begin{align}
\sum_{j = 1}^{\m'} \frac{d'_{j}}{\m' +\m} + \sum_{i =1}^{\m} \frac{d_{i}}{\m} &\leq 1, 
 \label{eq:dd_CSIT_outer_Kusers_1} \\
\sum_{j = 1}^{\m'} \frac{d'_{j}}{\m'} + \sum_{i =1}^{\m} \frac{d_{i}}{\m' + \m} &\leq 1, 
 \label{eq:dd_CSIT_outer_Kusers_2} \\
 d'_{j} &\leq 1,  \quad \forall j = 1, \ldots, \m',
\label{eq:dd_CSIT_outer_Kusers_3} \\ d_{i} &\leq 1 - \frac{1}{\T}, \quad \forall i = 1, \ldots, \m.
\label{eq:dd_CSIT_outer_Kusers_4}
\end{align}
\end{theorem}

\begin{proof}
The inequalities~\eqref{eq:dd_CSIT_outer_Kusers_3},~\eqref{eq:dd_CSIT_outer_Kusers_4} represent the single-user bounds on the static and the dynamic users, respectively~\cite{Telatar_capacity,Zheng_communication}. The remainder of the proof is dedicated to establishing the bounds~\eqref{eq:dd_CSIT_outer_Kusers_1} and~\eqref{eq:dd_CSIT_outer_Kusers_2}. 

We enhance the channel by providing global CSIR as well as allowing full cooperation among static users and full cooperation among dynamic users. 
 The enhanced channel is equivalent to a broadcast channel with two users: one static equipped with $\m'$ antennas and one dynamic equipped with $\m$ antennas. Define $\Y' \in \mathbb{C}^{\m'}$ and $\Y \in \mathbb{C}^{\m}$ to be the received signals of the static and the dynamic super-user, respectively, in the enhanced channel. We further enhance the channel by giving $\Y$ to the static user, generating a physically degraded channel since $\X \rightarrow ( \Y', \Y ) \rightarrow \Y$ forms a Markov chain. Feedback including delayed CSIT has no effect on capacity~\cite{Gamal_feedback}, therefore we remove it from consideration. Subsequently, we can utilize the  K{\"o}rner-Marton outer bound~\cite{Marton_coding}, 
\begin{align}
\sum_{j = 1}^{\m'} R'_j \leq & I(\X ; \Y', \Y | U, \Hset) \nonumber \\
\sum_{i = 1}^{\m} R_i \leq & I( U ; \Y | \Hset).
\end{align}
Therefore, from applying extremal entropy inequality~\cite{Yi_degrees,Liu_extremal,Weingarten_capacity_1} and Lemma~\ref{lemma:DoF_bounding},
\begin{align}
\sum_{j = 1}^{\m'} \twocolAlignMarker \frac{R'_{j}}{\m' + \m} + \sum_{i =1}^{\m} \frac{R_{i}}{\m} \twocolnewline
\leq & \frac{1}{\m' + \m} I(\X ; \Y', \Y | U, \Hset)
 + \frac{1}{\m} I( U ; \Y | \Hset) \nonumber \\
 = & \frac{1}{\m' + \m} h(\Y', \Y | U, \Hset) + o(\log(\rho)) \twocolbreak
 + \frac{1}{\m} h( \Y | \Hset) 
- \frac{1}{\m} h( \Y | U, \Hset) \nonumber \\
\leq & \log(\rho) + o(\log(\rho)). \label{eq:dd_CSIT_outer_Kusers_proof_2}
\end{align}
Therefore, the bound~\eqref{eq:dd_CSIT_outer_Kusers_1} is proved. Similarly, we can prove the bound~\eqref{eq:dd_CSIT_outer_Kusers_2} using the same steps after switching the roles of the two users in the enhanced channel.
\end{proof}

\section{Hybrid CSIT: Perfect CSIT for the Static Users and No CSIT for the Dynamic Users}
\label{section:pn_CSIT}

\begin{theorem}
\label{theorem:pn_CSIT_ach}
The fading broadcast channel characterized by Eq.~\eqref{eq:system_miso_bc}, with perfect CSIT for the static users and no CSIT for the dynamic users, can achieve the following multiuser degrees of freedom,
\begin{align}
\Dset_1 & : \sum_{j=1}^{\m'} \e^{\dagger}_j, \label{eq:pn_CSIT_ach_1} \\
\Dset_2, \ldots, \Dset_{\m+1} & : \frac{1}{\T} \sum_{j=1}^{\m'} \e^{\dagger}_j + (1 - \frac{1}{\T}) \e^{\dagger}_i, \quad i=1, \ldots, \m. \label{eq:pn_CSIT_ach_2} 
\end{align}
Therefore, their convex hull is also achievable.
\end{theorem}

\begin{proof}
$\Dset_1$ is achieved by inverting the channels of the static users at the transmitter and then every static user achieves one degree of freedom. $\Dset_{2}, \ldots, \Dset_{\m+1}$ in~\eqref{eq:pn_CSIT_ach_2} are achieved using product superposition along with channel inversion as follows. The transmitted signal over $\T$ instants is, 
\begin{equation}
\X = [\u, \;\; \u\v^{\dagger}],
\end{equation}
where $\u = \sum_{j=1}^{\m'} \b_j u_j$, $u_j$ is a symbol intended for the static user $j$, $\g_{j}^{\dagger} \b_j = 0$, and $\v \in \mathbb{C}^{\T - 1}$ contain independent symbols intended for the dynamic user $i$. Each of the static users receive an interference-free signal during the first time instant achieving one degrees of freedom. The dynamic user estimates its equivalent channel during the first time instant and decodes $\v$ during the remaining $(\T - 1)$ time instants.
\end{proof}

\begin{theorem}
\label{theorem:pn_CSIT_outer}
An outer bound on the degrees of freedom of the fading broadcast channel characterized by Eq.~\eqref{eq:system_miso_bc}, with perfect CSIT for the static users and no CSIT for the dynamic users, is 
\begin{align}
\sum_{j =1}^{\m'} \frac{d'_{j}}{\m' + 1 } + \sum_{i= 1}^{\m} d_{i} & \leq 1, \label{eq:pn_CSIT_outer_1} \\
d'_{j} & \leq 1, \quad \forall j = 1, \ldots, \m', 
\label{eq:pn_CSIT_outer_2} \\
 \sum_{i =1}^{\m} d_{i} & \leq 1 - \frac{1}{\T}. 
 \label{eq:pn_CSIT_outer_3}
\end{align}
\end{theorem}

\begin{proof}
The inequalities~\eqref{eq:pn_CSIT_outer_2} represent single-user bounds for the static users~\cite{Telatar_capacity}, and~\eqref{eq:pn_CSIT_outer_3} is a time-sharing outer bound for the dynamic users that was established in~\cite{Fadel_coherent, Fadel_disparity}. It remains to prove~\eqref{eq:pn_CSIT_outer_1}, as follows. 

We enhance the channel by giving global CSIR to all users and allowing full cooperation between the static users. This gives rise to an equivalent static user with  $\m'$ antennas receiving  $\Y'$ over an equivalent channel $\G$ and noise $\Z'$. At this point, we have a multi-user system where CSIT is available with respect to one user, but not others. We then bound the performance of this system with that of another (similar) system that has no CSIT. To do so, we use the {\em local statistical equivalence property} developed and used in~\cite{Mukherjee_vector,Tandon_fading,Tandon_synergistic}. First, we draw $\tilde{\G},\tilde{\Z}$ according to the distribution of $\G,\Z'$ and independent of them. We enhance the channel by providing $\tilde{\Y}=\tilde{\G}\X+\tilde{\Z}$ to the static receiver and $\tilde{\G}$ to all receivers. Because we do {\em not} provide $\tilde{\G}$ to the transmitter, there is no CSIT with respect to $\tilde{\Y}$. According to~\cite{Mukherjee_vector}, we have $h(\tilde{\Y},\Y'|\Hset) = h(\Y'|\Hset)+o(\log(\rho))$, where $\Hset = (\G,\tilde{\G}, \h_1, \ldots, \h_{\m})$, therefore we can remove $\Y'$ from the enhanced channel without reducing its degrees of freedom. This new equivalent channel has one user with $m'$ antennas receiving $(\tilde{\Y}, \Hset)$, $\m$ single-antenna users receiving $(y_i, \Hset)$, and no CSIT.\footnote{In the enhanced channel after removal of $\Y'$, the transmitter and receivers still share information about $\G$, but this random variable is now independent of all (remaining) transmit and receive variables.} Having no CSIT, the enhanced channel is in the form of a multilevel broadcast channel studied in Section~\ref{section:Ml_BC}, and hence using Theorem~\ref{theorem:KM_outer}, 

\begin{align}
\sum_{j=1}^{\m'} R'_{j} \leq & I( W ; \tilde{\Y} | U,\Hset ) 
 + I( \X; \tilde{\Y} | U, W, \Hset ) \nonumber \\
 R_1 \leq & I(U, W ; \y_{1}| V_1, \Hset ) 
 - I( W ; \tilde{\Y} | U,\Hset ) \nonumber \\
R_{i} \leq & I(V_{i-1} ; \y_{i}| V_i, \Hset ), \quad i = 2, \ldots, \m.
\end{align}
The dynamic receiver received signals have the same distribution. By following bounding steps parallel to~\eqref{eq:nn_CSIT_R_d_1},~\eqref{eq:nn_CSIT_R_d_2},~\eqref{eq:nn_CSIT_R_d_3},
\begin{align}
\sum_{j=1}^{\m} R_{i}
 \leq & \log(\rho) + o(\log(\rho)) 
- I(W ; \tilde{\Y} |U, \Hset) \twocolbreak
 - h( \y_{1} | U, W, \Hset).
\end{align}
Therefore,
\begin{align}
\sum_{j=1}^{\m'} \twocolAlignMarker \frac{R'_{j}}{\m' + 1} + \sum_{j=1}^{\m} R_{i} \twocolnewline
 \leq & \log(\rho) + o(\log(\rho)) \twocolbreak
+ (\frac{1}{\m' + 1} - 1 ) I(W; \tilde{\Y} |U, \Hset ) 
+ \frac{h( \tilde{\Y} | U, W, \Hset )}{\m' + 1}  \nonumber \\
& - h( \y_{1} | U, W, \Hset ), \\
\leq & \log(\rho) + o(\log(\rho)) \twocolbreak
 + \frac{h( \tilde{\Y}, \y_{1} | U, W, \Hset )}{\m' + 1} - h( \y_{1} | U, W, \Hset ) \\
\leq & \log(\rho) + o(\log(\rho)),
\end{align}
where the last inequality follows from applying the extremal entropy inequality~\cite{Yi_degrees,Liu_extremal,Weingarten_capacity_1} and Lemma~\ref{lemma:DoF_bounding}. This concludes the proof of the bound~\eqref{eq:pn_CSIT_outer_1}.
\end{proof}

\section{Hybrid CSIT: Perfect CSIT for the Static Users and Delayed CSIT for the Dynamic Users}
\label{section:pd_CSIT}

We begin with inner and outer bounds for one static and one dynamic user, then extend the result to multiple users. The transmitter knows the channel of the static users perfectly and instantaneously, and an outdated version of the channel of the dynamic users.

\subsection{Transmitting to One Static and One Dynamic User}
\label{section:pd_CSIT_one}

\begin{theorem}
\label{theorem:pd_CSIT_ach_one}
For the fading broadcast channel characterized by Eq.~\eqref{eq:system_miso_bc} with one static and one dynamic user, with perfect CSIT for the static user and delayed CSIT for the dynamic user, the achievable degrees of freedom region is the convex hull of the vectors,
\begin{align}
\Dset_1 : & (d',d) = (1 - \frac{1}{2\T} , \frac{1}{2} - \frac{1}{2\T}), \label{eq:pd_CSIT_ach_one_1} \\
\Dset_2 : & (d',d) = (\frac{1}{\T}, 1 - \frac{1}{\T} ). \label{eq:pd_CSIT_ach_one_2}
\end{align}
\end{theorem}

\begin{proof}
The degrees of freedom~\eqref{eq:pd_CSIT_ach_one_2} can be achieved by product superposition as discussed in Section~\ref{section:nn_CSIT}, without CSIT. We proceed to prove the achievability of~\eqref{eq:pd_CSIT_ach_one_1}. 

\begin{enumerate}
\item  Consider $[\u_1, \;\; \cdots, \;\; \u_{T-1}]$ to be a complex $2\times (T-1)$ matrix containing symbols intended for the static user, $[v_1, \;\; \cdots, \;\; v_{T-1}]$ intended for the dynamic user, and $\b \in \mathbb{C}$ is a beamforming vector so that $\g^\dagger \b=0$. In addition we define $\u_0={\bf 0}, v_0=1$. Using these components, the transmitter constructs and transmits a super-symbol of length $\T$, whose value at time $t$ is:
\begin{equation}
\x_1^{\dagger}(t) = \u_t+\b \; v_t.
\label{eq:first_trans_pd_csit}
\end{equation}
Note that $\x_1(0)=\b$ does not carry any information for either user, and serves as a pilot. The received super symbol at the static user is:
\begin{equation}
\y_1^{'\dagger} = [0, \;\; \g^{\dagger}\u_1, \;\; \cdots, \;\; \g^{\dagger}\u_{\T-1}].
\label{eq:first_static_pd_csit}
\end{equation}
The received super symbol at the dynamic user
\begin{align}
y_1^{\dagger} = & [ \h^{\dagger}_1\b, \;\; (\h^{\dagger}_1\u_1 + \h^{\dagger}_1\b v_1),\;\; \cdots,\twocolbreak \;\; (\h^{\dagger}_1\u_{\T-1} + \h^{\dagger}_1\b v_{\T-1}) ].
\label{eq:first_dynamic_pd_CSIT}
\end{align}
The dynamic user estimates its equivalent channel  $\h^{\dagger}_1\b$ from the received value in the first time instant. The remaining terms include symbols intended for the dynamic user plus some interference, whose cancellation is the subject of the next step.
 \item The transmitter next sends a second super symbol of length $\T$,
\begin{equation}
x_2 = [\bar{u},\;\; \bar{u}(\h^{\dagger}_1\u_1),\;\; \cdots,\;\; \bar{u}(\h^{\dagger}_1\u_{\T-1})],
\label{eq:second_trans_pd_CSIT}
\end{equation}
where $\bar{u} \in \mathbb{C}$ is a symbol intended for the static user. Hence,
\begin{equation}
y_2^{\dagger} = [ h_2\bar{u}, \;\; h_2\bar{u} (\h^{\dagger}_1\u_1),\;\;  \cdots, \;\; h_2\bar{u} (\h^{\dagger}_1\u_{\T-1})].
\label{eq:second_dynamic_pd_CSIT}
\end{equation}
The dynamic user estimates the equivalent channel $h_2\bar{u}$ during the first time instant and then acquires $\h^{\dagger}_1\u_t$, the interference in~\eqref{eq:first_dynamic_pd_CSIT}. Therefore, using $y_1, y_2$, the dynamic user solves for $v_t$ achieving $( \T - 1 )$ degrees of freedom. Furthermore,
\begin{equation}
y_2^{'\dagger} = [ g_1\bar{u}, \;\; g_1\bar{u}(\h^{\dagger}_1\u_1), \;\; \cdots, \;\; g_1\bar{u}(\h^{\dagger}_1\u_{\T-1})].
\label{eq:second_static_pd_CSIT}
\end{equation}
The static user solves for $\bar{u}$ achieving one degree of freedom and also uses $\h^{\dagger}_1\u_t$ to solve for $\u_t$ achieving further $2 \left( \T - 1\right)$ degrees of freedom.
\end{enumerate}
In summary, during $2\T$ instants, the static user achieves $(2\T - 1)$ degrees of freedom and the dynamic user achieves $(\T - 1)$ degrees of freedom. This shows the achievability of~\eqref{eq:pd_CSIT_ach_one_1}.
\end{proof}

\begin{theorem}
\label{theorem:pd_CSIT_outer_one}
For the fading broadcast channel characterized by Eq.~\eqref{eq:system_miso_bc} with one static and one dynamic user, where there is perfect CSIT for the static user and delayed CSIT for the dynamic user, an outer bound on the degrees of freedom region is,
\begin{align}
\frac{d'}{2} + d \leq & 1, \label{eq:pd_CSIT_outer_one_1} \\
 d' \leq & 1, \label{eq:pd_CSIT_outer_one_2} \\
 d \leq & 1 - \frac{1}{\T}. \label{eq:pd_CSIT_outer_one_3} 
\end{align}
\end{theorem}

\begin{proof}
The inequalities~\eqref{eq:pd_CSIT_outer_one_2} and~\eqref{eq:pd_CSIT_outer_one_3} represent the single-user outer bounds~\cite{Telatar_capacity,Zheng_communication}. It only remains to prove the outer bound~\eqref{eq:pd_CSIT_outer_one_1}, as follows. 
\begin{enumerate}
	\item We enhance the channel by giving global CSIR to both users and also give $y$ to the static user. The enhanced channel is physically degraded having $( \Y', \G )$ at the static user and $( y, \G )$ at the dynamic user, where $\Y' \triangleq (y', y)$ and $\G \triangleq (\h, \g)$. In a physically degraded channel, causal feedback (including delayed CSIT) does not affect capacity~\cite{Gamal_feedback}, so we can remove the delayed CSIT with respect to the dynamic user.
	\item We now use another enhancement with the motivation to remove the remaining CSIT (non-causal, with respect to the static user). This is accomplished, similar to Theorem~\ref{theorem:pn_CSIT_outer},  via local statistical equivalence property~\cite{Mukherjee_vector,Tandon_fading,Tandon_synergistic} in the following manner. We create a channel $\tilde{\G}$, and noise $\tilde{\Z}$ with the same distribution but independently of the true channel and noise, and a signal $\tilde{\Y}=\tilde{\G}\X+\tilde{\Z}$. A genie will give $\tilde{\Y}$ to the static receiver and $\tilde{\G}$ to both receivers. It has been shown~\cite{Mukherjee_vector} that $h(\tilde{\Y},\Y'|\Hset) = h(\Y'|\Hset)+o(\log\rho)$, where $\Hset = (\G,\tilde{\G})$, therefore we can remove $\Y'$ from the enhanced channel without reducing its degrees of freedom. 
	\item The enhanced channel is still physically degraded, therefore~\cite{Bergmans_random, Bergmans_simple}
\begin{align}
R' \leq & I ( \x ; \tilde{\Y} | U, \Hset ) 
 = h ( \tilde{\Y} | U, \Hset ) + o(\log(\rho)) \nonumber \\
R \leq & I ( U; y | \Hset ) 
 = h ( y | \Hset ) 
- h ( y | U, \Hset ),
\end{align}
where $U$ is an auxiliary random variable, and $U \rightarrow x \rightarrow (y', y)$ forms a Markov chain. Therefore, 
\begin{align}
\frac{1}{2} R' + R \leq &h ( y | \Hset ) 
 + \frac{1}{2} h ( \tilde{\Y} | U, \Hset ) \twocolbreak
 - h ( y | U, \Hset ) + o(\log(\rho)) \nonumber \\ 
 \leq & \log(\rho) + o(\log(\rho)),
\end{align}
where the last inequality follows from extremal entropy inequality and Lemma~\ref{lemma:DoF_bounding}~\cite{Yi_degrees,Liu_extremal,Weingarten_capacity_1}. This concludes the proof of  the bound~\eqref{eq:pd_CSIT_outer_one_1}.
\end{enumerate}
\end{proof}

\begin{remark}
\label{remark:pd_CSIT_one}
For the above broadcast channel with hybrid CSIT, the achievable sum degrees of freedom is $d_{\text{sum}} = \frac{3}{2} - \frac{1}{\T}$, and the outer bound on the sum degrees of freedom is $d_{\text{sum}} \leq \frac{3}{2}$. The gap decreases with the dynamic user coherence time (see Fig.~\ref{figure:pd_static_dynamic_one_T15} and~\ref{figure:pd_static_dynamic_one_T30}).
\end{remark}

\begin{figure}
\center
\includegraphics[width=\Figwidth]{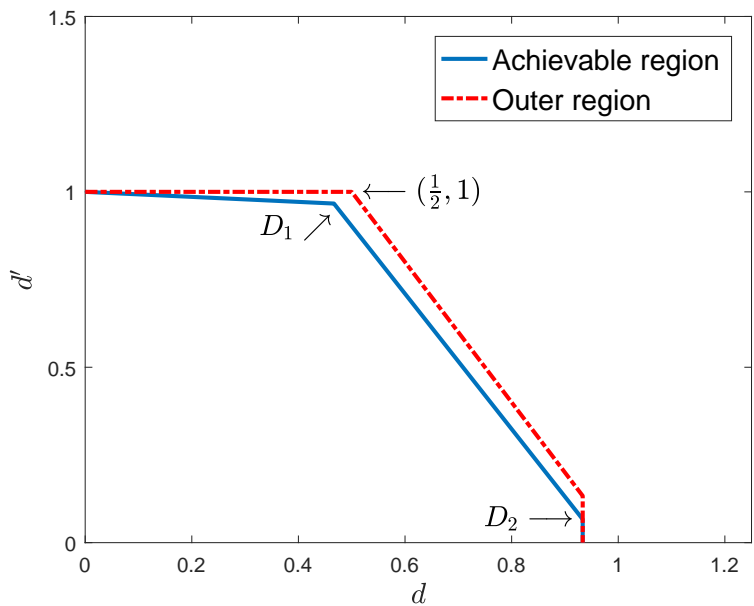}
\caption{One static and one dynamic user with hybrid CSIT and $\T = 15$}
\label{figure:pd_static_dynamic_one_T15}
\end{figure}
\begin{figure}
\center
\includegraphics[width=\Figwidth]{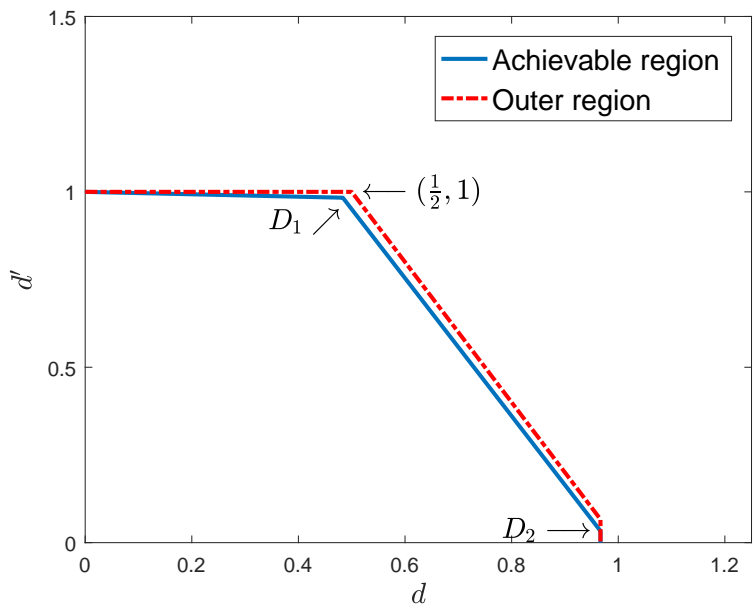}
\caption{One static and one dynamic user with hybrid CSIT and $\T = 30$}
\label{figure:pd_static_dynamic_one_T30}
\end{figure}

\subsection{Multiple Static and Dynamic Users}
\label{section:pd_CSIT_K}

\begin{theorem}
\label{theorem:pd_CSIT_ach_K}
The fading broadcast channel characterized by Eq.~\eqref{eq:system_miso_bc}, with perfect CSIT for the static users and delayed CSIT for the dynamic users, can achieve the following degrees of freedom,
\begin{align}
\Dset_1& : \sum_{j =1}^{\m'} \e^{\dagger}_j, \label{eq:pd_CSIT_ach_K_1} \\
\Dset_2, & \ldots, \Dset_{\m\m'+1}:  (1 - \frac{1}{2\T}) \e^{\dagger}_j + (\frac{1}{2} - \frac{1}{2\T}) \e^{\dagger}_i, 
\twocolbreak \quad j = 1, \ldots, \m', \ i= 1, \ldots, \m, \label{eq:pd_CSIT_ach_K_2} \\
\Dset_{\m\m'+2}, & \ldots, \Dset_{\m\m'+\m+2}: \frac{1}{\T} \sum_{j=1}^{\m'}\e^{\dagger}_j + (1 - \frac{1}{\T}) \e^{\dagger}_i, 
\twocolbreak \quad i = 1, \ldots, \m, \label{eq:pd_CSIT_ach_K_3} \\
\Dset_{\m\m'+\m+3}& : \frac{\m}{\T} \sum_{j=1}^{\m'} \e^{\dagger}_j 
\twocolbreak 
+ (\frac{1}{1 + \frac{1}{2} + \ldots + \frac{1}{\m}} (1 - \frac{\m}{\T})) \sum_{i=1}^{\m}\e^{\dagger}_i. 
\label{eq:pd_CSIT_ach_K_4}
\end{align}
The achievable region consists of the convex hull of the above vectors.
\end{theorem}

\begin{proof}
$\Dset_1$ is achieved by inverting the channel of the static users at the transmitter, providing one degree of freedom per static user. The achievability of $\Dset_{2}, \ldots, \Dset_{\m\m'+1}$ was established in Section~\ref{section:pd_CSIT_one}, and that of $\Dset_{\m\m'+2}, \ldots, \Dset_{\m\m'+\m+2}$ was proved in Section~\ref{section:pn_CSIT} without CSIT for the dynamic user, so it remains achievable with delayed CSIT. $\Dset_{\m\m'+\m+3}$ is achieved by retrospective interference alignment~\cite{Maddah_completely} along with product superposition, as follows. The transmitted signal over $\T$ instants is
\begin{equation}
\X = [ \bar{\U}, \;\; \bar{\U}\V],
\end{equation}
where $\bar{\U} \in \mathbb{C}^{\m \times \m}$ contains independent symbols intended for the static users sent by inverting the channels of the static users. Therefore, during the first $\m$ time instants, each static user receives an interference-free signal and achieves $\m$ degree of freedom, and furthermore the dynamic users estimate their equivalent channels. During the remaining time instants, each dynamic receiver obtains coherent observations of $(\T-\m)$ transmit symbols, which are pre-processed, combined and interference-aligned into super-symbols $\V$ according to retrospective interference alignment techniques of~\cite{Maddah_completely}. Accordingly, each dynamic receiver achieves $\frac{1}{1 + \frac{1}{2} + \ldots + \frac{1}{\m}} (1 - \frac{\m}{\T})$ degrees of freedom.
\end{proof}

\begin{theorem}
\label{theorem:pd_CSIT_outer_K}
An outer bound on the degrees of freedom region of the fading broadcast channel characterized by Eq.~\eqref{eq:system_miso_bc}, with perfect CSIT for the static users and delayed CSIT for the dynamic users, is
\begin{align}
\sum_{j = 1}^{\m'} \frac{d'_{j}}{\m' + \m} + \sum_{i = 1}^{\m} \frac{d_{i}}{\m} & \leq 1, 
\label{eq:pd_CSIT_outer_K_1} \\
\sum_{i = 1}^{\m} d_{i} & \leq \frac{\m}{ 1 + \frac{1}{2} + \ldots + \frac{1}{\m} }, 
 \label{eq:pd_CSIT_outer_K_2} \\
d'_{j} & \leq 1, \qquad j = 1, \ldots, \m',
 \label{eq:pd_CSIT_outer_K_3} \\
 d_{i} & \leq 1 - \frac{1}{\T}, \qquad i = 1, \ldots, \m.
 \label{eq:pd_CSIT_outer_K_4}
\end{align}
\end{theorem}

\begin{proof}
The inequalities~\eqref{eq:pd_CSIT_outer_K_3} and~\eqref{eq:pd_CSIT_outer_K_4} represent the single-user outer bounds for the static and dynamic users, respectively~\cite{Telatar_capacity,Zheng_communication}. According to Theorem~\ref{theorem:dd_CSIT_dynamic},~\eqref{eq:pd_CSIT_outer_K_2} represents an outer bound for the dynamic users. It only remains to prove~\eqref{eq:pd_CSIT_outer_K_1} as follows. 

\begin{enumerate}
	\item The original channel is enhanced by giving the users global CSIR. Furthermore, we assume full cooperation between the static users and between the dynamic users. The resulting enhanced channel is a broadcast channel with two users: one static user equipped with $\m'$ antennas, received signal $\Y'$, channel $\G$, and noise noise $\Z'$, and one dynamic user equipped with $\m$ antennas, received signal $\Y$, channel $\H$, and noise $\Z$. 
	\item We further enhance the channel by giving $\Y$ to the static user, constructing a physically degraded channel. For the enhanced channel, the static receiver is equipped with $\m' + \m$ antennas and has received signal $\hat{\Y}=[\Y^\dagger, \; \Y^{'\dagger}]^\dagger$, channel $\hat{\G}=[\G^\dagger, \; \H^\dagger]^\dagger$, and noise $\hat{\Z}=[\Z^\dagger, \; \Z^{'\dagger}]^\dagger$. Since any causal feedback (including delayed CSIT) does not affect the capacity of a physically degraded channel~\cite{Gamal_feedback}, the delayed CSIT for the dynamic receiver can be removed.

	\item  We now use another enhancement with the motivation to remove the remaining CSIT (non-causal, with respect to the static user). We create an artificial channel and noise, $\tilde{\G}$, $\tilde{\Z}$, with the same distribution but independent of $\hat{\G}$, $\hat{\Z}$, and a signal $\tilde{\Y}=\tilde{\G}\X+\tilde{\Z}$. A genie will give $\tilde{\Y}$ to the static receiver and $\tilde{\G}$ to both receivers. It has been shown~\cite{Mukherjee_vector} that $h(\tilde{\Y},\hat{\Y}|\Hset) = h(\hat{\Y}|\Hset)+o(\log\rho)$, where $\Hset = (\hat{\G},\tilde{\G})$, therefore we can remove $\hat{\Y}$ from the enhanced channel without reducing its degrees of freedom. 

\item The enhanced channel is physically degraded without CSIT, therefore~\cite{Bergmans_random, Bergmans_simple},
\begin{align}
\sum_{j=1}^{\m'} R'_{j} & \leq I( \X ; \tilde{\Y} | U, \Hset) \nonumber \\
\sum_{i=1}^{\m} R_{i} & \leq I( U ; \Y | \Hset).
\end{align}
Hence,
\begin{align}
\sum_{j=1}^{\m'} \frac{R'_{j}}{\m' + \m} + \twocolAlignMarker \sum_{j=1}^{\m} \frac{R_{i}}{\m} \twocolnewline
 \leq & \frac{1}{\m' + \m} h( \tilde{\Y} | U, \Hset) 
 + \frac{1}{\m} h( \Y | \Hset) \twocolbreak
 - \frac{1}{\m} h( \Y | U, \Hset)
 + o(\log(\rho)) \nonumber \\
\leq & \log(\rho) + o(\log(\rho)),
\end{align}
where the last inequality follows from the extremal entropy inequality~\cite{Yi_degrees,Liu_extremal,Weingarten_capacity_1} and Lemma~\ref{lemma:DoF_bounding} and since $h( \Y | \Hset) \leq \m \log(\rho) + o(\log(\rho))$~\cite{Telatar_capacity}. This concludes the proof of the bound~\eqref{eq:pd_CSIT_outer_K_1}.
\end{enumerate}
\end{proof}

\section{Conclusion}
\label{section:conclusion}
A multiuser broadcast channel was studied where some receivers experience longer coherence intervals and have CSIR while other receivers experience a shorter coherence interval and do not enjoy free CSIR. The degrees of freedom were studied under delayed CSIT, hybrid CSIT, and no CSIT. Among the techniques employed were interference alignment and beamforming along with product superposition for the inner bounds. The outer bounds involved a bounding of the rate region of the multiuser (discrete memoryless) multilevel broadcast channel. Some highlights of the results are: for one static and one dynamic user with delayed CSIT, the achievable degrees of freedom region partially meets the outer bound. For one static user with perfect CSIT and one dynamic user with delayed CSIT, the gap between the achievable and the outer sum degrees of freedom is inversely proportional to the dynamic user coherence time. For each of the considered CSI conditions, inner and outer bounds were also found for arbitrary number of users.

From these results we conclude that in the broadcast channel, coherence diversity delivers gains that are distinct from, and augment, the gains from beamforming and interference alignment. The authors anticipate that the tools and results of this paper can be helpful for future studies of hybrid CSIT/CSIR in other multi-terminal  networks.

\appendices
\section{Proof of Theorem~\ref{theorem:KM_outer}}
\label{appendix:KM_outer}

Recall, $\M'_{j}, \M_{i}$ are the messages of users $j = 1, \ldots, \m'$ and $i = 1, \ldots, \m$, respectively. We enhance the channel by assuming that user $j = 1, \ldots, \m'$ knows the messages $ M'_{j+1},\ldots,M'_{m'}$ and ${M}_{1}, \ldots,M_m$ and user $i = 1, \ldots, \m$ knows the messages ${M}_{i+1}, \ldots,{M}_{m} $. Using Fano's inequality, chain rule, and data processing inequality we can bound the rates of the static user $j = 1, \ldots, \m'$,
\begin{align}
n R'_{j} & \leq I( \M'_j ; Y'_{j,1}, \ldots, Y'_{j,n}| \M'_{j+1}, \ldots, \M'_{\m'}, \M_1, \ldots, \M_\m ) \\
 & = \sum_{k=1}^n I( \M'_j ; Y'_{j,k} | U_{j,k} ) \\ 
 & \leq \sum_{k=1}^n I( \M'_j, U_{j,k}, Y'_{j-1,1}, \ldots, Y'_{j-1,k-1} ; Y'_{j,k} | U_{j,k} ) \\
 & = \sum_{k=1}^n I( U_{j-1,k} ; Y'_{j,k} | U_{j,k} ) 
\label{eq:KM_outer_static} 
\end{align}
where 
\begin{equation*}
U_{j,k}= \big ( \M'_{j+1}, \ldots, \M'_{\m'}, \M_1, \ldots, \M_\m, Y'_{j,1}, \ldots, Y'_{j,k-1} \big),
\end{equation*}
$Y'_{j,k}$ denotes the received signal of user $j$ at time instant $k$, 
\begin{equation*}
U_{\m'} \rightarrow \cdots \rightarrow U_1 \rightarrow X \rightarrow ( Y'_1, \ldots, Y'_{\m'}, Y_1, \ldots, Y_{\m} )
\end{equation*}
forms a Markov chain, and $U_0 = X$. The rate of static user $\m'$ can be bounded as
\begin{align}
n R'_{\m'} \leq & \sum_{k=1}^n I( U_{\m'-1,k} ; Y'_{\m',k} | U_{\m',k} ) \\ 
 = & \sum_{k=1}^n I( X_k ; Y'_{\m',k} | U_{\m',k} ) 
 - \sum_{k=1}^n I( X_k ; Y'_{\m',k} | U_{\m'-1,k} ) \\
 \leq & \sum_{k=1}^n I( X_k, Y_{1,k+1}, \ldots, Y_{1,n} ; Y'_{\m',k} | U_{\m',k} ) \twocolbreak 
 - \sum_{k=1}^n I( X_k ; Y'_{\m',k} | U_{\m'-1,k} ) \\
 = & \sum_{k=1}^n I( Y_{1,k+1}, \ldots, Y_{1,n} ; Y'_{\m',k} | U_{\m',k} ) \twocolbreak 
 + \sum_{k=1}^n I( X_k ; Y'_{\m',k} | U_{\m',k}, Y_{1,k+1}, \ldots, Y_{1,n} ) \nonumber \\
 & - \sum_{k=1}^n I( X_k ; Y'_{\m',k} | U_{\m'-1,k} ) \\
 = & \sum_{k=1}^n I( W_k ; Y'_{\m',k} | U_{\m',k} ) \twocolbreak
 + \sum_{k=1}^n I( X_k ; Y'_{\m',k} | U_{\m',k}, W_k ) \twocolbreak
 - \sum_{k=1}^n I( X_k ; Y'_{\m',k} | U_{\m'-1,k} ),
\end{align}
where $W_k= Y_{1,k+1}^n$. Similarly,
\begin{align}
n R_{i} \leq & I( \M_i ; Y_{i,1}, \ldots, Y_{i,n} | \M_{i+1}, \ldots, \M_{\m} ) \\ 
 = & \sum_{k=1}^n I( \M_i ; Y_{i,k} | V_{i,k} ) \\
 = & \sum_{k=1}^n I( \M_i, V_{i,k}, Y_{i-1,k+1}, \ldots, Y_{i-1,n} ; Y_{i,k} | V_{i,k} ) \\ 
 = & \sum_{k=1}^n I( V_{i-1,k} ; Y_{i,k} | V_{i,k} ),
\end{align}
where we define $V_{i,k}\triangleq ( \M_{i+1}, \ldots, \M_{\m}, Y_{i,k+1}, \ldots, Y_{i,n} )$, which leads to the Markov chain
$V_{\m} \rightarrow \cdots \rightarrow V_1 \rightarrow (U_{\m'}, W) \rightarrow X \rightarrow ( Y'_1, \ldots, Y'_{\m'}, Y_1, \ldots,  Y_{\m} )$ . Using the chain rule and Csisz{\'a}r sum identity~\cite{Csiszar_information}, we obtain the bound~\eqref{eq:dynamic_2}.

\makeatletter
\if@twocolumn%
\begin{figure*}
\begin{align}
R_{1} \leq & \sum_{k=1}^n I( \M_1, \ldots, \M_\m ; Y_{1,k} | V_{1,k} ) \\
 \leq & \sum_{k=1}^n I( \M_1, \ldots, \M_\m, Y_{1,k+1}, \ldots, Y_{1,n} ; Y_{1,k} | V_{1,k} ) \\
 = & \sum_{k=1}^n I( \M_1, \ldots, \M_\m, Y_{1,k+1}, \ldots, Y_{1,n}, Y'_{\m',1}, \ldots, Y'_{\m',k-1} ; Y_{1,k} | V_{1,k} )\\ & - \sum_{k=1}^n I( Y'_{\m',1}, \ldots, Y'_{\m',k-1} ; Y_{1,k} | \M_1,\ldots, \M_\m, Y_{1,k+1}, \ldots, Y_{1,n} ) \\
 = & \sum_{k=1}^n I( U_{\m',k}, W_k ; Y_{1,k} | V_{1,k} )
 - \sum_{k=1}^n I( Y_{1,k+1}, \ldots, Y_{1,n} ; Y'_{\m',k} | U_{\m',k} ) \label{eq:dynamic_1}\\
 = & \sum_{k=1}^n I( U_{\m',k}, W_k ; Y_{1,k} | V_{1,k} ) 
 - \sum_{k=1}^n I( W_k ; Y'_{\m',k} | U_{\m',k} ).
\label{eq:dynamic_2}
\end{align}
\hrulefill
\end{figure*}
\else%
\begin{align}
R_{1} \leq & \sum_{k=1}^n I( \M_1, \ldots, \M_\m ; Y_{1,k} | V_{1,k} ) \\
 \leq & \sum_{k=1}^n I( \M_1, \ldots, \M_\m, Y_{1,k+1}, \ldots, Y_{1,n} ; Y_{1,k} | V_{1,k} ) \\
 = & \sum_{k=1}^n I( \M_1, \ldots, \M_\m, Y_{1,k+1}, \ldots, Y_{1,n}, Y'_{\m',1}, \ldots, Y'_{\m',k-1} ; Y_{1,k} | V_{1,k} )\\ & - \sum_{k=1}^n I( Y'_{\m',1}, \ldots, Y'_{\m',k-1} ; Y_{1,k} | \M_1,\ldots, \M_\m, Y_{1,k+1}, \ldots, Y_{1,n} ) \\
 = & \sum_{k=1}^n I( U_{\m',k}, W_k ; Y_{1,k} | V_{1,k} )
 - \sum_{k=1}^n I( Y_{1,k+1}, \ldots, Y_{1,n} ; Y'_{\m',k} | U_{\m',k} ) \label{eq:dynamic_1}\\
 = & \sum_{k=1}^n I( U_{\m',k}, W_k ; Y_{1,k} | V_{1,k} ) 
 - \sum_{k=1}^n I( W_k ; Y'_{\m',k} | U_{\m',k} ).
\label{eq:dynamic_2}
\end{align}
\fi%
\makeatother

By introducing a time-sharing auxiliary random variable, $Q$,~\cite{Gamal_network} and defining
\begin{align}
X \triangleq& ( X_Q , Q ),& Y'_j &\triangleq ( Y'_{j,Q}, Q )\nonumber \\
Y_i \triangleq& ( Y_{i,Q}, Q ),& U_i &\triangleq ( U_{i,Q}, Q )\nonumber \\
V_j \triangleq& ( V_{j,Q}, Q ),& W &\triangleq ( W_Q, Q ),
\end{align}
we establish~\eqref{eq:KM_outer_1}-\eqref{eq:KM_outer_4}. Similarly, we can follow the same steps to prove~\eqref{eq:KM_outer_5}-\eqref{eq:KM_outer_8} after switching the role of the two sets of variables $Y_1',\ldots,Y_{\m'}'$ and $Y_1,\ldots,Y_{\m}$. This completes the proof of Theorem~\ref{theorem:KM_outer}.

\section{Multilevel Broadcast Channel with Degraded Message Sets}
\label{appendix:degraded_message}

Here, we study the capacity of the multiuser multilevel broadcast channel that is characterized by~\eqref{KM_outer_MC} with degraded message sets. In particular, $\M_0 \in \left[1:2^{n \R_0}\right]$ is to be communicated to all receivers, and furthermore $\M_1 \in \left[1:2^{n \R_1}\right]$ is to be communicated to receiver $Y_1$.\footnote{For compactness of expression, here we refer to each receiver by the variable denoting its received signal.} A three-receiver special case was studied by Nair and El Gamal~\cite{Nair_capacity} where the idea of indirect decoding was introduced, and the capacity is the set of rate pairs $\left( \R_1, \R_0 \right)$ such that
\begin{align}
  \R_0 \leq & \min\big\{ I( U  ; Y_2 ) , I( V ; Y'_1  ) \big\} , \nonumber \\
  \R_1 \leq & I( X ; Y_1  | U )  , \nonumber \\
\R_0  + \R_1 \leq & I( V  ; Y'_1 ) + I( X ; Y_1 | V ) ,
\end{align}
for some pmf $p(u, v)p(x|v)$. In the sequel, we give a generalization of Nair and El Gamal for multiuser multilevel broadcast channel.

\begin{theorem}
\label{Theorem:common}
The capacity of multiuser multilevel broadcast channel characterized by~\eqref{KM_outer_MC}, with degraded message sets, is the set of rate pairs $( \R_1, \R_0 )$ such that
\begin{align}
  \R_0 \leq & \min\big\{ I( U  ; Y_\m  ) , I( V ; Y'_{\m'}  ) \big\}, \nonumber \\
  \R_1 \leq & I( X ; Y_1  | U )  , \nonumber \\
\R_0  + \R_1 \leq & I( V  ; Y'_{\m'}  ) + I( X ; Y_1 | V ) ,
\end{align}
for some pmf $p(u, v)p(x|v)$.
\end{theorem}

\begin{proof}
The converse parallels the proof of the converse of the three-receiver case studied by Nair and El Gamal in~\cite{Nair_capacity} after replacing $Y_2, Y'_1$ with $Y_\m, Y'_{\m'}$, respectively. In particular, $U$ and $V$ are defined as follows. 
\begin{align*}
U_k &\triangleq ( \M_0, Y_{1,1}, \ldots, Y_{1,k-1}, Y_{\m,k+1}, \ldots, Y_{\m,n} ), \\
V_k & \triangleq ( \M_0, Y_{1,1}, \ldots, Y_{1,k-1}, Y'_{\m',k+1}, \ldots, Y'_{\m',n} ),
\end{align*}
$k=1,\ldots, n$, and let $Q$ be a time-sharing random variable uniformly distributed over the set $\{1, \ldots, n \}$ and independent of $X^n, Y_1^n, Y_{\m,1}, \ldots, Y_{\m,n}, Y'_{\m', 1}, \ldots, Y'_{\m', n}$. We then set $U=( U_Q, Q ), V= ( V_Q, Q), X=X_Q, Y_1=Y_{1,Q}, Y_\m=Y_{\m,Q}$, and $Y'_{\m'}=Y'_{\m',Q}$. This completes the converse part of the proof.

The achievability part uses superposition coding and indirect decoding 
as follows.
\begin{itemize}
  \item \textbf{Rate splitting:} divide the private message $\M_1$ into two independent messages $\M_{10}$ at rate $\R_{10}$ and $\M_{11}$ at rate $\R_{11}$, where 
  $\R_1= \R_{10} + \R_{11}$.
	
  \item \textbf{Codebook generation:} fix a pmf $p(u,v)p(x|v)$ and randomly and independently generate $2^{n\R_0}$ sequences $u^n\left(m_0\right)$, $m_0 \in \left[1: 2^{n\R_0}\right]$, each according to $\prod_{k=1}^n p_U(u_k)$. For each $m_0$, randomly and conditionally independently generate $2^{n\R_{10}}$ sequences $v^n (m_0, m_{10} )$, $m_{10} \in [1: 2^{n\R_{10}} ]$, each according to $\prod_{k=1}^n p_{V|U}(v_k|u_k(m_0) )$. For each pair $(m_0, m_{10})$, randomly and conditionally independently generate $2^{n\R_{11}}$ sequences $x^n (m_0, m_{10}, m_{11} )$, $m_{11} \in [1: 2^{n\R_{11}} ]$, each according to $\prod_{k=1}^n p_{X|V} (x_k|v_k(m_0, m_{10} ))$.
  
\item \textbf{Encoding:} to send the message pair $(m_0, m_1 )=(m_0, m_{10}, m_{11})$, the encoder transmits $x^n (m_0, m_{10}, m_{11} )$.

\item \textbf{Decoding at the users $Y_2, \ldots, Y_\m$:} decoder $i$ declares that $\hat{m}_{0i} \in [ 1 : 2^{n \R_0} ]$ is sent if it is the unique message such that $( u^n (\hat{m}_{0i} ), y_i^n ) \in \Tset_{\epsilon}^{(n)}$. Hence, by law of large numbers and the packing lemma~\cite{Gamal_network}, the probability of error tends to zero as $n \rightarrow \infty$ if 
\begin{align}\label{eq:Ach_common_k}
 \R_0 &< \min_{\quad 2 \leq i \leq \m}\{ I ( U  ; Y_i ) - \delta(\epsilon) \},  \nonumber \\ 
 &=  I( U  ; Y_\m ) - \delta(\epsilon)  , 
\end{align}
where the last equality follows from applying data processing inequality on the Markov chain 
$U \rightarrow X \rightarrow Y_1 \rightarrow Y_2 \rightarrow \cdots \rightarrow Y_\m$.

\item \textbf{Decoding at $Y_1$:} decoder 1 declares that $( \hat{m}_{01}, \hat{m}_{10}, \hat{m}_{11} )$ is sent if it is the unique message triple such that $\big( u^n ( \hat{m}_{01} ), v^n ( \hat{m}_{01}, \hat{m}_{10} ), x^n (\hat{m}_{01}, \hat{m}_{10}, \hat{m}_{11} ), y_1^n \big) \in  [ 1 : 2^{n \R_0} ]$. Hence, by law of large numbers and the packing lemma~\cite{Gamal_network}, the probability of error tends to zero as $n \rightarrow \infty$ if 
\begin{align}\label{eq:Ach_commonandpriv_1}
  \R_{11} & < I ( X  ; Y_1 | V ) - \delta(\epsilon), \nonumber \\
 \R_{10} + \R_{11} & < I ( X  ; Y_1 | U ) - \delta(\epsilon), \nonumber \\
\R_0 + \R_{10} + \R_{11} & < I ( X  ; Y_1 ) - \delta(\epsilon).
  \end{align}

  \item \textbf{Decoding at users $Y'_1, \ldots, Y'_{\m'}$:} decoder $j$ decodes $m_0$ indirectly by declaring $\tilde{m}_{0j}$ is sent if it is the unique message such that $( u^n (\tilde{m}_{0j} ), v^n (\tilde{m}_{0j}, m_{10} ), z_j^n ) \in \Tset_{\epsilon}^{(n)}$ for some $m_{10} \in [ 1 : 2^{n \R_0}]$. Hence, by law of large numbers and packing lemma, the probability of error tends to zero as $n\rightarrow \infty$ if 
\begin{align}\label{eq:Ach_common_j}
 \R_0 + \R_{10} &< \min_{\quad 1 \leq j \leq \m'} \{ I ( U, V ; Y'_j ) - \delta(\epsilon) \}, \nonumber\\
  &= \min_{\quad 1 \leq j \leq \m'} \{ I ( V ; Y'_j ) - \delta(\epsilon) \}, \nonumber\\
  &=  I ( V ; Y'_{\m'} ) - \delta(\epsilon),
  \end{align}
where the last two equalities follow from applying the chain rule and data processing inequality on the Markov chain $U \rightarrow V \rightarrow X \rightarrow Y'_1 \rightarrow Y'_2 \rightarrow \cdots \rightarrow Y'_{\m'}$.
\end{itemize}
By combining the bounds in \eqref{eq:Ach_common_k}, \eqref{eq:Ach_commonandpriv_1}, \eqref{eq:Ach_common_j}, substituting $\R_{10}+\R_{11}=\R_1$, and eliminating $\R_{10}$ and $\R_{11}$ by the Fourier-Motzkin procedure~\cite{Nair_capacity}, the proof of the achievability is completed.
\end{proof}

\bibliographystyle{IEEEtran}
\bibliography{IEEEabrv,References}

\begin{thebibliography}{10}
\providecommand{\url}[1]{#1}
\csname url@rmstyle\endcsname
\providecommand{\newblock}{\relax}
\providecommand{\bibinfo}[2]{#2}
\providecommand\BIBentrySTDinterwordspacing{\spaceskip=0pt\relax}
\providecommand\BIBentryALTinterwordstretchfactor{4}
\providecommand\BIBentryALTinterwordspacing{\spaceskip=\fontdimen2\font plus
\BIBentryALTinterwordstretchfactor\fontdimen3\font minus
  \fontdimen4\font\relax}
\providecommand\BIBforeignlanguage[2]{{%
\expandafter\ifx\csname l@#1\endcsname\relax
\typeout{** WARNING: IEEEtran.bst: No hyphenation pattern has been}%
\typeout{** loaded for the language `#1'. Using the pattern for}%
\typeout{** the default language instead.}%
\else
\language=\csname l@#1\endcsname
\fi
#2}}

\bibitem{Fadel_block}
M.~Fadel and A.~Nosratinia, ``Block-fading broadcast channel with hybrid {CSIT}
  and {CSIR},'' in \emph{IEEE International Symposium on Information Theory
  (ISIT)}, June 2017, pp. 1873--1877.

\bibitem{Huang_degrees}
C.~Huang, S.~Jafar, S.~Shamai, and S.~Vishwanath, ``On degrees of freedom
  region of {MIMO} networks without channel state information at
  transmitters,'' \emph{{IEEE} Trans. Inf. Theory}, vol.~58, no.~2, pp.
  849--857, Feb. 2012.

\bibitem{Vaze_degree}
C.~Vaze and M.~Varanasi, ``The degree-of-freedom regions of {MIMO} broadcast,
  interference, and cognitive radio channels with no {CSIT},'' \emph{{IEEE}
  Trans. Inf. Theory}, vol.~58, no.~8, pp. 5354--5374, Aug. 2012.

\bibitem{Lapidoth_capacity}
A.~Lapidoth, S.~Shamai, and M.~Wigger, ``On the capacity of fading {MIMO}
  broadcast channels with imperfect transmitter side-information,'' \emph{arXiv
  preprint cs/0605079}, 2006.

\bibitem{Jafar_blind}
S.~Jafar, ``Blind interference alignment,'' \emph{{IEEE} J. Sel. Topics Signal
  Process.}, vol.~6, no.~3, pp. 216--227, June 2012.

\bibitem{Fadel_broadcast}
M.~Fadel and A.~Nosratinia, ``Broadcast channel under unequal coherence
  intervals,'' in \emph{IEEE International Symposium on Information Theory
  (ISIT)}, July 2016, pp. 275--279.

\bibitem{Fadel_disparity}
------, ``Coherence disparity in broadcast and multiple access channels,''
  \emph{{IEEE} Trans. Inf. Theory}, vol.~62, no.~12, pp. 7383--7401, Dec. 2016.

\bibitem{Caire_achievable}
G.~Caire and S.~Shamai, ``On the achievable throughput of a multiantenna
  {Gaussian} broadcast channel,'' \emph{{IEEE} Trans. Inf. Theory}, vol.~49,
  no.~7, pp. 1691--1706, July 2003.

\bibitem{Weingarten_capacity}
H.~Weingarten, Y.~Steinberg, and S.~Shamai, ``The capacity region of the
  {Gaussian} multiple-input multiple-output broadcast channel,'' \emph{{IEEE}
  Trans. Inf. Theory}, vol.~52, no.~9, pp. 3936--3964, Sept. 2006.

\bibitem{Davoodi_aligned}
A.~Davoodi and S.~Jafar, ``Aligned image sets under channel uncertainty:
  Settling a conjecture by {Lapidoth}, {Shamai} and {Wigger} on the collapse of
  degrees of freedom under finite precision {CSIT},'' \emph{arXiv preprint
  arXiv:1403.1541}, 2014.

\bibitem{Maddah_completely}
M.~Maddah-Ali and D.~Tse, ``Completely stale transmitter channel state
  information is still very useful,'' \emph{{IEEE} Trans. Inf. Theory},
  vol.~58, no.~7, pp. 4418--4431, July 2012.

\bibitem{Gou_optimal}
T.~Gou and S.~Jafar, ``Optimal use of current and outdated channel state
  information: Degrees of freedom of the {MISO} {BC} with mixed {CSIT},''
  \emph{{IEEE} Commun. Lett.}, vol.~16, no.~7, pp. 1084--1087, July 2012.

\bibitem{Tandon_fading}
R.~Tandon, M.~A. Maddah-Ali, A.~Tulino, H.~V. Poor, and S.~Shamai, ``On fading
  broadcast channels with partial channel state information at the
  transmitter,'' in \emph{International Symposium on Wireless Communication
  Systems (ISWCS)}, Aug. 2012, pp. 1004--1008.

\bibitem{Amuru_degrees}
S.~Amuru, R.~Tandon, and S.~Shamai, ``On the degrees-of-freedom of the 3-user
  {MISO} broadcast channel with hybrid {CSIT},'' in \emph{IEEE International
  Symposium on Information Theory (ISIT)}, 2014, pp. 2137--2141.

\bibitem{Tandon_synergistic}
R.~Tandon, S.~Jafar, S.~Shamai, and V.~Poor, ``On the synergistic benefits of
  alternating {CSIT} for the {MISO} broadcast channel,'' \emph{{IEEE} Trans.
  Inf. Theory}, vol.~59, no.~7, pp. 4106--4128, July 2013.

\bibitem{Li_product}
Y.~Li and A.~Nosratinia, ``Product superposition for {MIMO} broadcast
  channels,'' \emph{{IEEE} Trans. Inf. Theory}, vol.~58, no.~11, pp.
  6839--6852, Nov. 2012.

\bibitem{Li_coherent}
------, ``Coherent product superposition for downlink multiuser {MIMO},''
  \emph{{IEEE} Trans. Wireless Commun.}, vol.~PP, no.~99, pp. 1--9, 2014.

\bibitem{Fadel_coherent}
M.~Fadel and A.~Nosratinia, ``Coherent, non-coherent, and mixed--{CSIR}
  broadcast channels: {Multiuser} degrees of freedom,'' in \emph{IEEE
  International Symposium on Information Theory (ISIT)}, June 2014, pp.
  2574--2578.

\bibitem{Fadel_coherence}
------, ``Coherence disparity in time and frequency,'' in \emph{Proc. {IEEE}
  Global Telecommunication Conference ({GLOBECOM}'16)}, Dec. 2016, pp. 1--6.

\bibitem{Zhang_spatially}
F.~Zhang, M.~Fadel, and A.~Nosratinia, ``Spatially correlated {MIMO} broadcast
  channel: {Analysis} of overlapping correlation eigenspaces,'' in \emph{IEEE
  International Symposium on Information Theory (ISIT)}, June 2017, pp.
  1097--1101.

\bibitem{Borade_multilevel}
S.~Borade, L.~Zheng, and M.~Trott, ``Multilevel broadcast networks,'' in
  \emph{IEEE International Symposium on Information Theory (ISIT)}, June 2007,
  pp. 1151--1155.

\bibitem{Nair_capacity}
C.~Nair and A.~Gamal, ``The capacity region of a class of three-receiver
  broadcast channels with degraded message sets,'' \emph{{IEEE} Trans. Inf.
  Theory}, vol.~55, no.~10, pp. 4479--4493, Oct. 2009.

\bibitem{Liu_extremal}
T.~Liu and P.~Viswanath, ``An extremal inequality motivated by multiterminal
  information-theoretic problems,'' \emph{{IEEE} Trans. Inf. Theory}, vol.~53,
  no.~5, pp. 1839--1851, May 2007.

\bibitem{Liu_vector}
R.~Liu, T.~Liu, V.~Poor, and S.~Shamai, ``A vector generalization of {Costa's}
  entropy-power inequality with applications,'' \emph{{IEEE} Trans. Inf.
  Theory}, vol.~56, no.~4, pp. 1865--1879, Apr. 2010.

\bibitem{Marzetta_capacity}
T.~Marzetta and B.~Hochwald, ``Capacity of a mobile multiple-antenna
  communication link in {Rayleigh} flat fading,'' \emph{{IEEE} Trans. Inf.
  Theory}, vol.~45, no.~1, pp. 139--157, Jan. 1999.

\bibitem{Zheng_communication}
L.~Zheng and D.~Tse, ``Communication on the {Grassmann} manifold: a geometric
  approach to the noncoherent multiple-antenna channel,'' \emph{{IEEE} Trans.
  Inf. Theory}, vol.~48, no.~2, pp. 359--383, Feb. 2002.

\bibitem{Marton_coding}
K.~Marton, ``A coding theorem for the discrete memoryless broadcast channel,''
  \emph{{IEEE} Trans. Inf. Theory}, vol.~25, no.~3, pp. 306--311, May 1979.

\bibitem{Yang_degrees_1}
S.~Yang, M.~Kobayashi, D.~Gesbert, and X.~Yi, ``Degrees of freedom of time
  correlated {MISO} broadcast channel with delayed {CSIT},'' \emph{{IEEE}
  Trans. Inf. Theory}, vol.~59, no.~1, pp. 315--328, Jan. 2013.

\bibitem{Yi_degrees}
X.~Yi, S.~Yang, D.~Gesbert, and M.~Kobayashi, ``The degrees of freedom region
  of temporally correlated {MIMO} networks with delayed {CSIT},'' \emph{{IEEE}
  Trans. Inf. Theory}, vol.~60, no.~1, pp. 494--514, Jan. 2014.

\bibitem{Shannon_zero}
C.~Shannon, ``The zero error capacity of a noisy channel,'' \emph{{IEEE} Trans.
  Inf. Theory}, vol.~2, no.~3, pp. 8--19, Sept. 1956.

\bibitem{Bergmans_random}
P.~Bergmans, ``Random coding theorem for broadcast channels with degraded
  components,'' \emph{{IEEE} Trans. Inf. Theory}, vol.~19, no.~2, pp. 197--207,
  Mar. 1973.

\bibitem{Bergmans_simple}
------, ``A simple converse for broadcast channels with additive white
  {Gaussian} noise,'' \emph{{IEEE} Trans. Inf. Theory}, vol.~20, no.~2, pp.
  279--280, Mar. 1974.

\bibitem{Gamal_feedback}
A.~E. Gamal, ``The feedback capacity of degraded broadcast channels
  (corresp.),'' \emph{{IEEE} Trans. Inf. Theory}, vol.~24, no.~3, pp. 379--381,
  May 1978.

\bibitem{Telatar_capacity}
E.~Telatar, ``Capacity of multi-antenna {Gaussian} channels,'' \emph{European
  transactions on telecommunications}, vol.~10, no.~6, pp. 585--595, 1999.

\bibitem{Weingarten_capacity_1}
H.~Weingarten, T.~Liu, S.~Shamai, Y.~Steinberg, and P.~Viswanath, ``The
  capacity region of the degraded multiple-input multiple-output compound
  broadcast channel,'' \emph{{IEEE} Trans. Inf. Theory}, vol.~55, no.~11, pp.
  5011--5023, Oct. 2009.

\bibitem{Mukherjee_vector}
P.~Mukherjee, R.~Tandon, and S.~Ulukus, ``Secure degrees of freedom region of
  the two-user {MISO} broadcast channel with alternating {CSIT},'' \emph{{IEEE}
  Trans. Inf. Theory}, vol.~PP, no.~99, pp. 1--1, Apr. 2017.

\bibitem{Csiszar_information}
I.~Csisz{\'a}r and J.~K{\"o}rner, ``Information theory : {Coding} theorems for
  discrete memoryless channels,'' \emph{Budapest: Akad{\'e}miai Kiad{\'o}},
  1981.

\bibitem{Gamal_network}
A.~E. Gamal and Y.~Kim, \emph{Network information theory}.\hskip 1em plus 0.5em
  minus 0.4em\relax Cambridge University Press, 2011.

\end{thebibliography}

\end{document}